\documentclass[12pt]{article}%
\usepackage{amssymb}
\usepackage{amsmath}
\usepackage{sw20elba}
\usepackage{graphicx}
\usepackage{numinsec}
\usepackage{amsfonts}%
\providecommand{\U}[1]{\protect\rule{.1in}{.1in}}
\newtheorem{theorem}{Theorem}[section]

\newenvironment{proof}[1][Proof]{\textbf{#1.} }{\ \rule{0.5em}{0.5em}}
\ifx\pdfoutput\relax\let\pdfoutput=\undefined\fi
\newcount\msipdfoutput
\ifx\pdfoutput\undefined\else
\ifcase\pdfoutput\else
\ifx\paperwidth\undefined\else
\ifdim\paperheight=0pt\relax\else\pdfpageheight\paperheight\fi
\ifdim\paperwidth=0pt\relax\else\pdfpagewidth\paperwidth\fi
\fi\fi\fi
\begin{document}

\author{Alexander Figotin and Guillermo Reyes\\Department of Mathematics\\University of California at Irvine\\Irvine, CA 92697-3875}
\title{Multi-transmission-line-beam interactive system}
\maketitle

\begin{abstract}
We construct here a Lagrangian field formulation for a system consisting of an
electron beam interacting with a\ slow-wave structure modeled by a possibly
non-uniform multiple transmission line (MTL). In the case of \ a single line
we recover the linear model of a traveling wave tube (TWT) due to J.R. Pierce.
Since a properly chosen MTL can approximate a real waveguide structure with
any desired accuracy, the proposed model can be used in particular for design
optimization. Furthermore, the Lagrangian formulation provides for: (i) a
clear identification of the mathematical source of amplification, (ii) exact
expressions for the conserved energy and its flux distributions obtained from
the Noether theorem. In the case of uniform MTLs we carry out an exhaustive
analysis of eigenmodes and find sharp conditions on the parameters of the
system to provide for amplifying regimes.

\end{abstract}

\section{Introduction}

We study here theoretical aspects of generation and amplification of microwave
(millimeter waves) radiation by traveling wave tubes. Generally speaking,
generation and amplification of electromagnetic radiation can be produced by
enormous variety of devices of different designs depending on the frequency of
the radiation and its power. For light such devices are lasers; we remind that
the term laser is an acronym for Light Amplification by Stimulated Emission of
Radiation. For microwaves, which are of our special interest here, there is a
large class of amplifying devices including maser, a predecessor of the laser,
magnetrons, klystrons, traveling wave tubes, crossed-field amplifiers and gyrotrons.

In the case of lasers, as suggested by its very name, the general principle of
the amplification is based on the stimulated emission resulting from certain
atomic transitions. "Lasers come in a great variety of forms, using many
different laser materials, many different atomic systems, and many different
kinds of pumping or excitation techniques. The beams of radiation that lasers
emit or amplify have remarkable properties of directionality, spectral purity,
and intensity.", \cite[p. 2]{Siegman}. An important and defining property of
laser radiation is its coherency, that is its monochromaticity. For
amplification the coherency means that in a narrow frequency band the output
signal, after being amplified, reproduces pretty accurately the shape of the
input signal but with a substantial increase in amplitude. Coherent
amplification combined with a feedback allows to produce highly directional
and highly monochromatic beams. Observe that atomic transitions of the laser
medium constitute a fundamental basis of amplification, that is the
amplification mechanism is fixed by the nature, so to speak. There is extended
literature on the theory of lasers, see for instance \cite[4.1]{Fox},
\cite{Loudon}, \cite{Siegman}. Its basic phenomenological elements include:
(i) Einstein's treatment of the spontaneous and stimulated emission,
\cite[4.1]{Fox}, and (ii) operation principle based on interaction between the
laser (gain) medium and electromagnetic modes of a cavity containing this
medium. More detailed and fundamental theory that can justify the laser
phenomenology involves quantum optics (electronics), \cite{Loudon}, \cite{Fox}.

In the case of microwaves the radiation is produced by microwave vacuum
electronic devices, known formerly as \emph{microwave tubes}. These devices
use free electrons in a vacuum to convert energy from a DC power source to an
RF (radio frequency) signal. In other words, as a result of interaction
between the electron beam and properly designed structure the kinetic energy
of the electrons is converted into electromagnetic energy stored in the field,
\cite{Gilm1}, \cite[2.2]{Nusinovich}, \cite[4]{SchaB}, \cite{Tsimring}.
\emph{The key operational principle of any microwave device is a positive
feedback interaction between coherent radiation by electrons radiating in
phase on one hand and on the other hand electron bunching caused by radiation
on the stream of electrons. The electron bunching associated with acceleration
and deceleration of groups of electrons along the beam constitutes the
physical mechanism of radiation generation and its amplification.}

An important class of microwave devices uses as its operation principle the
\emph{Cherenkov radiation} generated by charged particles propagating in or
near a medium supporting slow waves with phase velocity comparable with the
particle velocity. Traveling wave tubes, the main subject of our studies here,
belongs to this class.

Traveling wave tubes (TWT) are used widely in many areas including satellite
communication and radar systems. Typical TWT consists of an elongated vacuum
tube containing an electron beam which passes down the middle of an RF circuit
(a \emph{slow-wave structure}). The operation principle of a TWT is as
follows. At one end of the TWT structure, the RF circuit is fed with a
low-powered radio signal to be amplified. As the RF signal travels along the
tube at near the same speed as the electron beam, the electromagnetic field
acts upon the beam and causes electron bunching with consequent formation of
the so-called \emph{space-charge wave}. The electromagnetic field associated
with the space-charge wave induces more current back into the RF circuit, thus
enhancing the bunching, and so on. The EM field thus builds up and is
amplified as it passes down the structure until a saturation regime is reached
and a large RF signal is collected at the output. The role of the slow-wave
structure is to slow down the electromagnetic wave to match up with the
velocity of the electrons in the beam, usually a small fraction of the speed
of light. Such a synchronism is required for effective in phase interaction
between the structure and the beam with optimal extraction of the kinetic
energy of the electrons. A typical slow-wave structure is the helix, which
reduces the speed of propagation according to its pitch. Further details on
the design and operation of TWT can be found in \cite{Gilm1}, \cite{PierTWT},
\cite{Tsimring}, \cite[4]{Nusinovich}.

An effective mathematical model for a TWT interacting with an electron beam
was introduced by J. R. Pierce, \cite{PierTWT}, \cite[I]{Pier51}. This model
is the simplest one that accounts for wave amplification along the structure,
energy extraction from the electron beam and its conversion into microwave
radiation in the TWT, see also \cite[4]{SchaB}, \cite{Gilm1}, \cite{Gilm},
\cite{Tsimring} and \cite[4]{Nusinovich}. In Section \ref{SectionPierce}, we
provide for precise description of the model as presented in \cite[I]{Pier51}.
The mentioned presentation is a time domain model, in contrast to other
presentations dealing with the frequency domain counterpart. Though simple,
the Pierce model allows for adequate estimates of the gain and it was used
effectively in designing working TWTs in the fifties. This model captures
remarkably well significant features of wave amplification and the beam-wave
energy transfer, and is still in use for basic design estimates.

The model presented by Pierce is one-dimensional and\ consists of (i) an ideal
linear representation of the electron beam and (ii) a lossless transmission
line (TL) representing the waveguide structure. The transmission line is
assumed to be homogeneous, that is, with uniformly distributed capacitance and
inductance. To overcome the Pierce theory limitations far more sophisticated
nonlinear theories have been developed to model very involved physics of the
electron beam and slow-wave structures, \cite{SchaB}, \cite{Gilm},
\cite{Tsimring}. Needless to say that those theories are far more complex and
often require a massive computer work.

In this paper we advance the Pierce theory to a theory that, while keeping its
simplicity and constructiveness, allows for more complex slow-wave structures.
We start by developing a Lagrangian field framework for the original Pierce
model. Such framework allows for extension of the model in two directions:
\textit{a)} we can replace the transmission line by a multi-transmission line
(MTL) and \textit{b)} we can dispense with the homogeneity assumption, thus
considering general nonhomogeneous systems consisting of a multi-transmission
line (MTL) coupled to an electron beam. We refer to such a system as a MTLB
system. Extension to multiple transmission lines is motivated by the fact that
general MTLs can approximate with desired accuracy real waveguided structures
which can be homogeneous (uniform) as well as inhomogeneous (nonuniform),
\cite{Nitsch}, \cite{Paul}, \cite{SchwiE}.

One of the advantages of the Lagrangian formulation is that conservation laws
and explicit expressions for the conserved quantities and their fluxes can be
obtained at once from the Noether theorem. We would like to point out that
though conservation laws do follow from the Euler-Lagrange evolution equations
there is no systematic way to extract them from those equations. In addition
to that, since all the information about dynamics is encoded in the scalar
Lagrange function we can trace the amplification mechanisms and the properties
of the energy transfer from the electron beam to the microwave radiation to
certain terms in the Lagrangian density.

For homogeneous MTLB systems, we study the amplification phenomenon by
considering the exponentially growing eigenmodes and associated complex-valued
wave numbers for the field equations, just as in the original Pierce theory,
\cite{PierTWT}, \cite{PierW}, \cite{Pier51}. We provide also a rigorous proof
of the fact that, on the growing mode, the energy always flows in the expected
direction, \textit{i.e.} from the beam to the MTL. In this case, the
eigenmodes analysis can be carried out analytically, providing for explicit
expressions for their energy density and energy flux distributions as well as
sufficient conditions for the existence of amplification regimes (growing
modes). The analysis includes derivation of a special canonical form of the
dispersion relation having a remarkable feature: one of its two terms depends
only on the MTL, whereas another one depends only the beam parameters. Such a
special factorization and separation of variables simplifies the analysis significantly.

As to inhomogeneous MTLB systems, they are by far more involved compared to
homogeneous ones. In particular, for periodic MTLB systems the dispersion
relations are not polynomial and that requires to turn to the most general
form of the Floquet theory, \cite[II, III]{YakSta1}. For this general case we
provide the first step towards a systematic study, namely we transform the
Euler-Lagrange field equations\ into the canonical Hamiltonian form using
basics of the de Donder-Weyl theory, \cite[4.2]{Rund}. This particular
Hamiltonian form consists of a system of equations which is of first order in
the spatial variable, thus providing the basis for the most effective use of
the Floquet theory, \cite[II, III]{YakSta1} in the study of periodic
structures. Detailed development of the Floquet theory for periodic MTLB
system requires to overcome a number of technical difficulties and it is left
for future studies.

One of the features of the proposed here phenomenological approach is that it
captures the electron bunching as a physical mechanism of amplification in
some form. Consequently, our analysis is a valuable source of a solid
information on the electron bunching.

The structure of the paper is as follows: In Section \ref{SectMainResults} we
briefly summarize our main results. Section \ref{SectionPierce} is devoted to
the description of Pierce's model for beam-TL interaction as presented in
\cite[I]{Pier51}. Section \ref{SectLagrangian} deals with the Lagrangian
approach to the model, including generalizations to both non-homogeneous and
multiple transmission lines. In the following Section
\ref{SectAmplificationGeneral copy(1)}, we explore the amplification mechanism
in the MTLB system as linked to instabilities in the dynamics of the beam. The
appropriate mathematical setting, in particular the Hamiltonian structure of
the model aimed at the study of eigenmodes in the periodic case is the subject
of Section \ref{SectHamiltonian}. In Section \ref{AmplMTL-beam}, we focus on
the detailed study of growing modes for the homogeneous MTLB system. Section
\ref{EnergyConsEx} deals with the questions of \ general energy conservation
and energy transfer between the beam and the MTL on the growing mode. In
Section \ref{PierceRev copy(1)} we make apparent how our general approach
allows to easily recover some of the original Pierce's results.

Finally, in Section \ref{MathSubj} we collect some technically involved
subjects which have been deferred there to avoid distracting the reader from
the main flow of ideas.

\section{Main results\label{SectMainResults}}

One of the goals of this work is to identify the mathematical mechanism of
amplification in MLTB systems. This goal has been accomplished by the
construction of a Lagrangian field theory of MLTB systems that underlines
their physical properties. Leaving detailed developments of this theory to the
following sections we simply identify here the key term of the system
Lagrangian responsible for amplification. This term quite expectedly is
associated with the electron beam and is described by the following expression%
\begin{equation}
\mathcal{L}_{\mathrm{b}}=\frac{\xi}{2}\left(  \partial_{t}q+u_{0}\partial
_{z}q\right)  ^{2},\qquad\xi=\frac{4\pi}{\omega_{\mathrm{p}}^{2}\sigma}>0,
\label{potenergyterm}%
\end{equation}
where $t$ and $z$ are, respectively, time and longitudinal variable,
$q=q(t,z)$ is the charge ("smoothed-out jelly of charge", \cite[I]{Pier51})
flowing through the beam. $\sigma$ and $u_{0}$ stand, respectively, for the
cross section and the electron velocity and $\omega_{\mathrm{p}}$ is the
plasma frequency. According to the general theory, we can identify the kinetic
and potential energies of the beam by expanding the expression
(\ref{potenergyterm}), that is
\[
\mathcal{L}_{\mathrm{b}}=\frac{\xi}{2}\left(  \partial_{t}q\right)  ^{2}+\xi
u_{0}\partial_{t}q\partial_{z}q+\frac{\xi}{2}u_{0}^{2}\left(  \partial
_{z}q\right)  ^{2},
\]
where \emph{the potential energy of the beam }$-\frac{\xi}{2}\left(
u_{0}\partial_{z}q\right)  ^{2}$\emph{ is a negative quantity. This is a
marked feature distinguishing MLTB from common oscillatory systems, in which
the potential energy is always positive}. \emph{The negative sign of this
potential energy term is ultimately responsible for system instability and
consequent amplification}.

Indeed, a typical oscillatory system has a positive potential energy
manifested in forces that move the system toward its equilibrium state. The
simplest examples are given by a linear mass-spring system or its electric
analog - a simple electric $LC$ oscillatory circuit. The corresponding
Lagrangians are%
\[
\mathcal{L}_{\mathrm{1}}(x,x^{\prime})=\frac{1}{2}mx^{\prime2}-\frac{1}%
{2}kx^{2};\qquad\mathcal{L}_{\mathrm{2}}(q,q^{\prime})=\frac{1}{2}Lq^{\prime
2}-\frac{1}{2C}q^{2},
\]
where $m$ is the mass of the point, $k$ is the elastic Hooke constant of the
spring and $q$ is the charge in the capacitor. Such forces result in a stable
motion with oscillatory energy transfer between its kinetic and potential
forms. A qualitatively different picture occurs when the potential energy is
negative, as in $\mathcal{L}_{\mathrm{b}}.$ In this case resulting forces move
the system away from the equilibrium at an exponentially growing rate. Such
situation corresponds to having a negative Hooke constant $k$ in
$\mathcal{L}_{\mathrm{1}}$ or a negative capacitance in $\mathcal{L}%
_{\mathrm{2}}$ above. Interestingly, Pierce has observed an effective negative
capacitance in his studies of a transmission line interacting with the
electron beam, \cite{Pier51}.

Another marked feature of the term $\mathcal{L}_{\mathrm{b}}$ in
(\ref{potenergyterm}) is its degeneration as quadatric form manifested as a
perfect square trinomial expression or, alternatively, as a precise gyrotropic
term. According to the general theory of unstable regimes, \cite{YakSta1},
this kind of degeneration is a necessary condition for instability arising
under proper perturbations. From the point of view of the second order partial
differential equation describing the beam dynamics this degeneracy is
manifested as parabolicity compared to hyperbolicity occurring for common wave motion.

The power and efficiency of the Lagrangian approach is further demonstrated by
an exhaustive analysis of amplification regimes for a general homogeneous MTLB
system, including precise conditions under which amplification takes place. In
particular, if $\ 0\leq v_{1}\leq v_{2}\leq...\leq v_{n}$ \ denote the
characteristic velocities of the MTL as an independent system, we show that
there is always an amplifying regime if $u_{0}\leq v_{1}$. If $u_{0}>v_{1}$,
we show that amplification occurs only for sufficiently small $\xi$ in
(\ref{potenergyterm}). We also provide a transparent form of the dispersion
relation for a general homogeneous MTLB system, including possible
degenerations, as well as an asymptotic analysis of the amplification factor
as $\ $the beam parameter $\xi$ defined in (\ref{potenergyterm}) becomes
arbitrarily small or large. The limits $\xi$ $\rightarrow0$ and $\xi
\rightarrow\infty$ correspond to high, respectively small electron density of
the beam. In \cite{Pier51}, Pierce deals with large values of $\xi,$ which
allows him to simplify the dispersion relation to an exactly solvable third
degree equation for the forward eigenmodes. We review Pierce's result in the
light of our approach.

Yet another benefit of our Lagrangian approach is an exhaustive analysis of
the energetic issues, including the overall energy conservation and energy
transfer between the MTL and the beam. This analysis yields explicit
expressions for the power $P_{\mathrm{B}\rightarrow\mathrm{MTL}}$ flowing from
the beam to the MTL for an exponentially growing solution of the form%
\begin{equation}
Q(z,t)=\widehat{Q}\mathrm{e}^{-\mathrm{i(}\omega t-k_{0}z)},\qquad
q(z,t)=\widehat{q}\mathrm{e}^{-\mathrm{i(}\omega t-k_{0}z)},\text{\qquad
}\operatorname{Im}k_{0}<0,
\end{equation}
where $Q$ is the coordinate describing the MTL and $q$ is the one describing
the beam. Namely, the following formula holds%
\begin{equation}
\left\langle P_{\mathrm{B}\rightarrow\mathrm{MTL}}\right\rangle (z)=-\left[
\omega\xi\left\vert k_{0}\right\vert ^{2}\left\vert \widehat{q}\right\vert
^{2}(\operatorname*{Re}v_{0}-u_{0})\operatorname{Im}v_{0}\right]
\mathrm{e}^{-2\left(  \operatorname{Im}k_{0}\right)  z},\text{\qquad}%
v_{0}=\frac{\omega}{k_{0}}. \label{IntroAmplFormula}%
\end{equation}
We show that in the above formula the constant in front of the exponential is
indeed \emph{positive,} meaning that the energy flows from the beam to the
MTL. Formula (\ref{IntroAmplFormula}) indicates also that the power
transferred to the MTL increases exponentially in the direction of the
electron flow. The opposite is true of the evanescent wave when the power
flows to the beam and decreases exponentially in the $+z$ direction.

\subsection{Negative potential energy and general gain media}

Looking at the above analysis we can identify two main features of the
Lagrangian providing for the amplification in the MTLB system. The first one
is the fact that the beam potential energy is negative and unbounded from
below. This feature of the electron beam Lagrangian clearly indicates that the
model is an ideal one with the negative potential energy term representing
effectively an inexhaustible source of energy. This energy can be converted
into another form of energy such as energy of electromagnetic radiation. Such
ideal model can be suitable for describing the amplification and gain up to
the point of saturation. The saturation can conceivably be modeled
phenomenologically by introducing an additional positive potential energy term
into the beam Lagrangian represented by a higher order polynomial with a small
coefficient. That would make the theory nonlinear, of course.

The second feature of the Lagrangian providing for amplification is a
particular degeneracy of the expression (\ref{potenergyterm}) for the
Lagrangian and its role in the system stability. More precisely, such term
makes the system unstable under proper perturbations, as discussed in detail
in Section \ref{ChgWave}. It is a well known fact from the Floquet theory of
periodic Hamiltonian systems that such degeneration is indeed necessary in
order to have unstable perturbations, \cite{YakSta1}

The association of the amplification and gain with a negative potential energy
term in a system Lagrangian can be a general way to model gain media.
Interestingly, the phenomenon of negative energy waves in inhomogeneous
plasmas is well known and understood at phenomenological level, see for
instance, \cite[7.7]{Bellan}, \cite[1.3]{Hasegawa}, \cite[3.1]{Melrose}. The
explanation provided in the cited references is essentially that in the
approximate phenomenological model the wave-energy density corresponds to the
change in the total system energy density in a more detailed theory. Such
negative energy waves typically occur when the system is near equilibrium with
a steady-state flow velocity and there exists a mode that reduces the average
kinetic energy of the particles to a value below the initial equilibrium
value. Importantly, concepts of negative energy waves and gain media are
intimately related to the instability.

It is instructive to compare and contrast the developed here approach for
modeling the gain medium by a negative potential energy term with the
conventional approach that represent the gain medium as a system with negative
absorption. As an important and relevant example of later let us consider
colisionless plasma in a weak external electric field $E=E_{0}e^{-i(\omega
t-kx)}$ described in \cite[3]{LiP}. The interactions in such a plasma are
non-local and consequently the plasma permittivity depends on the both on
$\omega$ and $k$ (the so-called spatial dispersion), and it has non-zero
imaginary part resulting in dissipation. The imaginary part of permittivity
and the energy dissipation are given respectively by formulas%
\begin{equation}
\epsilon^{\prime\prime}=-\frac{4\pi^{2}e^{2}m}{k^{2}}\left[  \frac{\partial
f}{\partial p}\right]  _{v=\omega/k};\qquad Q=\frac{\omega}{8\pi}%
\epsilon^{\prime\prime}\left\vert E\right\vert ^{2}=-\left\vert E\right\vert
^{2}\frac{\pi me^{2}\omega}{2k^{2}}\left[  \frac{\partial f}{\partial
p}\right]  _{v=\omega/k}, \label{LandauDis}%
\end{equation}
where $m,e$ are the electron mass (respectively charge) and $f$ is the
momentum distribution function of the stationary plasma. If the plasma is
isotropic (that is, the distribution function of momenta only depends on
$\left\vert p\right\vert $ or, in the one-dimensional case, is an even
function), it can be shown that $Q>0$, (\cite[3.30]{LiP}), and consequently
the plasma absorbs energy from the field, a phenomenon called \textsl{Landau
damping}. However, in the presence of anisotropy, the sign of $\left[
\frac{\partial f}{\partial p}\right]  _{v=\omega/k}$ and hence that of $Q$
might be reversed yielding a net flow of energy from the electrons to the
field and providing an example of gain medium. It is intuitively clear from
(\ref{LandauDis}) that the net energy flux depends on the relative number of
electrons with the velocity larger/smaller than the phase velocity of the wave.

Main differences\ between our approach for modeling gain in the MTLB system
and the conventional approach for modeling gain in the plasma example
described above are as follows. The conventional approach is based
fundamentally on the concept of open system and the gain medium is not modeled
explicitly but rather by its effect on the system. In our approach the beam
interacting with the electric field form a conservative system and the gain
medium is modeled explicitly as the beam term in the system Lagrangian with a
negative potential component. Yet another difference is that, in the MTLB
system, the gain occurs for the space charge wave velocities larger or smaller
than the wave phase velocity.

In fact, a causal dissipative system can always be extended uniquely to a
properly constructed conservative system, \cite{FigSch1}, \cite{FigShi1},
\cite{FigSch2}. It is an interesting question then whether one can carry out
similar construction for the gain medium. Answering this question is not in
the scope of this paper but we intend to look at this subject in our future work.

\section{Pierce's model\label{SectionPierce}}

In \cite[I]{Pier51}, J.R. Pierce presented a linear, one-dimensional model for
the description of the \ interaction of an electron beam with a surrounding
waveguide. The model is based on the following assumptions.

\textbf{Assumption I}. \textit{The modulation of both the electron velocity
and the current on the beam (so called a.c. components) are small compared to
the average or unperturbed velocity and current}.

This assumption justifies the linearization of the equations around the
unperturbed regime. Let the total velocity of the electrons be $u_{0}+v,$
where $u_{0\text{ }}$is the average velocity and $v$ is a small perturbation.
Analogously, let $\rho_{0}+\rho$ be the total electron density (per unit
volume) where $\rho_{0}$ is the unperturbed density and $\rho$ is the
perturbation. Let $\sigma$ be the cross section of the beam. Then, the total
current flowing is $I_{T}=I_{0}+I_{\mathrm{b}}$, where $I_{0}=\sigma\rho
_{0}u_{0}$ is the d.c. current and the perturbation is given by%
\begin{equation}
I_{\mathrm{b}}=\sigma\left(  \rho u_{0}+v\rho_{0}+\rho v\right)  .
\label{beam1}%
\end{equation}
Linearization around the d.c. regime gets rid of the term $\rho v,$ which is
quadratic in the perturbations. Thus we take%
\begin{equation}
I_{\mathrm{b}}=\sigma\left(  \rho u_{0}+v\rho_{0}\right)  \label{beam1bis}%
\end{equation}
in what follows. The linearized conservation of charge equation reads%
\begin{equation}
\frac{\partial\rho}{\partial t}+\frac{\partial i}{\partial z}=\frac
{\partial\rho}{\partial t}+\frac{1}{\sigma}\frac{\partial I_{\mathrm{b}}%
}{\partial z}=0, \label{beambis}%
\end{equation}
where $t$ represents time, $z$ is the longitudinal variable and $i$ is the
current density, $i=I_{\mathrm{b}}/\sigma$.

\textbf{Assumption II}. \textit{The beam is thought of as a continuous medium
(electron jelly) with no internal stress and a unique volumetric force acting
along it, namely the one resulting from the axial component of the electric
field associated to the signal on the waveguide.}

It is further assumed that the charge/mass ratio in the electron jelly is
precisely $e/m,$ $e=-\left\vert e\right\vert $ \ being the electron charge and
$m$ being the electron mass. Therefore, if $E=E_{z}$ \ is the axial component
of the field, the motion equation for the medium reads:%
\begin{equation}
\frac{\partial v}{\partial t}+(u_{0}+v)\frac{\partial v}{\partial z}=\frac
{e}{m}E, \label{beam3}%
\end{equation}
where, on the left-hand side, we have used the usual Eulerian expression for
the acceleration in terms of the velocity field $v(z,t).$ Upon linearization,
the term $v\frac{\partial v}{\partial z}$ is dropped, thus yielding%
\begin{equation}
\frac{\partial v}{\partial t}+u_{0}\frac{\partial v}{\partial z}=\frac{e}{m}E.
\label{beam3bis}%
\end{equation}
Notice that in Pierce's original paper, \cite[I]{Pier51}, the charge of the
electron is denoted by $-e,$ whereas here it is just $e$.

Actually, the full blown Pierce model, as presented in the book \cite{PierTWT}%
, also includes the effect of electron-electron repulsion in the beam (so
called space charge effects); see also \cite{Gilm},\cite{Tsimring}. Here we do
not include such effect for the sake of simplicity, but we advance that this
can be done and we plan to report on this issue in the future.

Taking the derivatives of (\ref{beam1bis}) with respect to $t$ and $z$ \ we
obtain the following expressions for $\partial v/\partial t$ and $\partial
v/\partial z:$%
\begin{equation}
\frac{\partial v}{\partial t}=\frac{1}{\sigma\rho_{0}}\frac{\partial I_{b}%
}{\partial t}-\frac{u_{0}}{\rho_{0}}\frac{\partial\rho}{\partial t}%
;\qquad\frac{\partial v}{\partial z}=\frac{1}{\sigma\rho_{0}}\frac{\partial
I_{b}}{\partial z}-\frac{u_{0}}{\rho_{0}}\frac{\partial\rho}{\partial z}.
\label{beamfdef}%
\end{equation}
We use (\ref{beambis}) \ to express $\partial\rho/\partial t$ in terms of
$\partial I_{b}/\partial z$ in the first of the above relations and
differentiate the resulting relation with respect to $t$ thus yielding%
\begin{equation}
\frac{\partial^{2}v}{\partial t^{2}}=\frac{1}{\sigma\rho_{0}}\frac
{\partial^{2}I_{b}}{\partial t^{2}}+\frac{u_{0}}{\sigma\rho_{0}}\frac
{\partial^{2}I_{b}}{\partial z\partial t}. \label{equpierce1}%
\end{equation}
Next, we differentiate the second relation in (\ref{beamfdef}) with respect to
$t,$ expressing again $\partial\rho/\partial t$ in terms of $\partial
I_{b}/\partial z$ . We obtain%
\begin{equation}
\frac{\partial^{2}v}{\partial z\partial t}=\frac{1}{\sigma\rho_{0}}%
\frac{\partial^{2}I_{b}}{\partial z\partial t}+\frac{u_{0}}{\sigma\rho_{0}%
}\frac{\partial^{2}I_{b}}{\partial z^{2}}. \label{equpierce2}%
\end{equation}
On the other hand, differentiating (\ref{beam3bis}) with respect to $t$ \ we
get%
\begin{equation}
\frac{\partial^{2}v}{\partial t^{2}}+u_{0}\frac{\partial^{2}v}{\partial
z\partial t}=\frac{e}{m}\frac{\partial E}{\partial t}. \label{equpierce3}%
\end{equation}
Finally, we replace the second derivatives in (\ref{equpierce3}) through their
expressions in (\ref{equpierce1}) and (\ref{equpierce2}), yielding a second
order equation for the beam current%
\begin{equation}
\partial_{t}^{2}I_{\mathrm{b}}+2u_{0}\partial_{t}\partial_{z}I_{\mathrm{b}%
}+u_{0}^{2}\partial_{z}^{2}I_{\mathrm{b}}=\sigma\frac{e}{m}\rho_{0}%
\partial_{t}E \label{beam4}%
\end{equation}
(here and in what follows, we use $\partial_{t}^{2}$ for $\partial/\partial
t^{2},$ $\partial_{z}^{2}$ for $\partial/\partial z^{2},$ etc. for the sake of
brevity). Next, Pierce considers the reciprocal action of the electron beam on
the transmission line (TL).

\textbf{Assumption III}. \textit{The action of the beam onto the waveguide
amounts to a shunt current instantaneously induced on the line. This current
is equal in absolute value and opposite to the current on the beam.}

According to this assumption, the usual transmission line (telegraph)
equations are modified so as to include an additional source term,
\cite[I]{Pier51},
\begin{equation}
\partial_{z}I=-C\partial_{t}V-\partial_{z}I_{\mathrm{b}},\qquad\partial
_{z}V=-L\partial_{t}I.\label{traneq1}%
\end{equation}
Here, as usual, $I=I\left(  t,z\right)  $ and $V=V\left(  t,z\right)  $ denote
respectively the current through the inductive element and the voltage on the
shunt capacitive element of the TL, $C>0$ and $L>0$ are respectively the shunt
capacitance and inductance per unit of length. Note also that in the equations
(\ref{traneq1}) $\partial_{z}I$ and $\partial_{z}V$ are respectively the
current through the shunt capacitive element and the voltage drop on the
inductive element of the TL per unit length. The addition of the source term
$-\partial_{z}I_{\mathrm{b}}$ can be justified under the assumption of
quasi-stationarity of the process: the charge wave on the beam "mirrors" onto
the line. One of the lumped elements in the discretization of such excited TL
is represented in Fig. \ref{Circuit}. Induced current can be thought of as a
distributed shunt current source.%
%TCIMACRO{\FRAME{ftbpFU}{4.2843in}{2.2935in}{0pt}{\Qcb{Discrete element of the
%TL-beam system in Pierce's model. The arrows represent shunt current induced
%on the capacitor.}}{\Qlb{Circuit}}{circuitonuevo.eps}%
%{\special{ language "Scientific Word";  type "GRAPHIC";
%maintain-aspect-ratio TRUE;  display "USEDEF";  valid_file "F";
%width 4.2843in;  height 2.2935in;  depth 0pt;  original-width 4.2333in;
%original-height 2.2537in;  cropleft "0";  croptop "1";  cropright "1";
%cropbottom "0";  filename 'CircuitoNuevo.eps';file-properties "XNPEU";}} }%
%BeginExpansion
\begin{figure}[ptb]%
\centering
\ifcase\msipdfoutput
\includegraphics[
height=2.2935in,
width=4.2843in
]%
{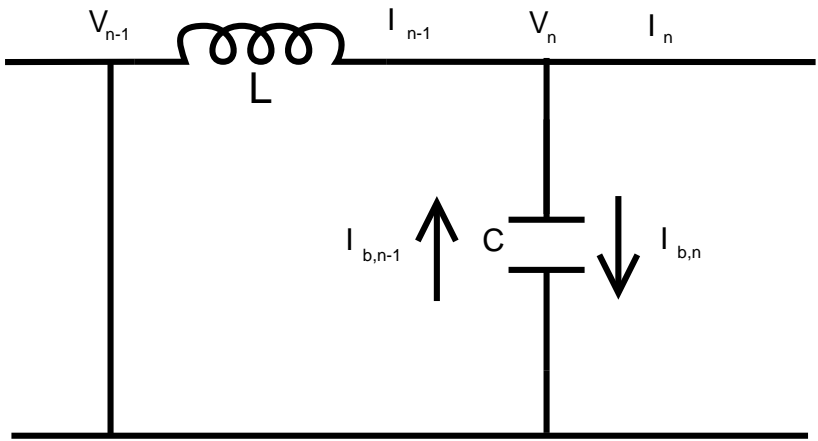}%
\else
\includegraphics[
height=2.2935in,
width=4.2843in
]%
{D:/alatex/Preparation/Reyes/lagsys/arxiv2/graphics/CircuitoNuevo__1.pdf}%
\fi
\caption{Discrete element of the TL-beam system in Pierce's model. The arrows
represent shunt current induced on the capacitor.}%
\label{Circuit}%
\end{figure}
%EndExpansion

The axial component of the electric field associated to the waveguide is
related to the TL voltage:
\begin{equation}
E\left(  t,z\right)  =-\partial_{z}V\left(  t,z\right)  . \label{traneq3}%
\end{equation}
Plugging the above expression into (\ref{beam4}), we arrive at the equation%
\begin{equation}
\partial_{t}^{2}I_{\mathrm{b}}+2u_{0}\partial_{t}\partial_{z}I_{\mathrm{b}%
}+u_{0}^{2}\partial_{z}^{2}I_{\mathrm{b}}=-\sigma\frac{e}{m}\rho_{0}%
\partial_{t}\partial_{z}V. \label{traneq4}%
\end{equation}
Thus, according to \cite[I]{Pier51}, the equations (\ref{traneq1}) and
(\ref{traneq4}) constitute a model of the interactive TL-beam (TLB) system.

Some comments are in order. In more recent literature, improved versions of
the linear Pierce model have been considered, see e.g. \cite[4]{Nusinovich}.
These versions account for finer features such as bunching saturation, or
retain the nonlinearity present in the original versions of equations
(\ref{beam1}) and (\ref{beam3}), etc. Although such enriched models are
undoubtedly more realistic and numerical computations based on them might
provide a better agreement with experiment, they hardly allow for analytical
treatment. In particular, they do not possess a Lagrangian structure. Pierce's
model, though simple, already captures the mechanism of amplification and, as
mentioned in the Introduction, can be generalized to the case of MTLB systems,
and allows for a thorough mathematical analysis in all cases. Taking into
account the fact that real wave guides can be approximated, in principle, by
an MTL with any degree of accuracy, \cite{Nitsch}, \cite{Paul}, \cite{SchwiE},
such generalization opens new perspectives in design optimization, which is
the ultimate goal of our study.

\section{Lagrangian formulation of Pierce's model\label{SectLagrangian}}

In this section we construct a Lagrangian field theory underlying the Pierce
model. The Lagrangian theory provides a deeper insight into mathematical
mechanism of amplification and energy transfer from the electron beam to the radiation.

\subsection{The Lagrangian\label{Lagrangian}}

The linear system of equations (\ref{traneq1})-(\ref{traneq4}) arises as
Euler-Lagrange equations associated to certain quadratic Lagrangian. To see
this, let us first introduce the charge variables $Q$ and $q$ related
respectively to the currents $I$ and $I_{\mathrm{b}}$ by%
\begin{equation}
I=\partial_{t}Q,\qquad I_{\mathrm{b}}=\partial_{t}q. \label{tranbe1}%
\end{equation}
Thus the variables $Q,q$ represent the amount of charge traversing the
cross-section of the line (respectively the beam) at the point $z$ within the
time interval $(t_{0},t),$ where $t_{0}$ is some fixed reference time. Then
the TLB system (\ref{traneq1}) and (\ref{traneq4}) takes the form%
\begin{equation}
\partial_{z}Q=-CV-\partial_{z}q,\qquad\partial_{z}V=-L\partial_{t}^{2}Q,
\label{tranbe2}%
\end{equation}%
\begin{equation}
\left(  \partial_{t}+u_{0}\partial_{z}\right)  ^{2}q=-\frac{\sigma
\omega_{\mathrm{p}}^{2}}{4\pi}\partial_{z}V, \label{tranbe3}%
\end{equation}
where $\omega_{\mathrm{p}}$ is the \emph{plasma frequency} defined (in
Gaussian units) by%
\begin{equation}
\omega_{\mathrm{p}}^{2}=\frac{4\pi e\rho_{0}}{m}, \label{tranbe4}%
\end{equation}
\cite[2.2]{DavNP}.

Since it is not any harder to deal with inhomogeneous (in particular,
periodic) TLs, we suppose from now on that $C$ and $L$ can be position
dependent, that is%
\begin{equation}
C=C\left(  z\right)  ,\qquad L=L\left(  z\right)  . \label{tranbe5}%
\end{equation}
Notice that the first equation in (\ref{tranbe2}) readily implies the
following representation for $V$
\begin{equation}
V=-C^{-1}\partial_{z}(Q+q). \label{tranbe6}%
\end{equation}
Inserting the above expression for $V$ into the second equation in
(\ref{tranbe2}) and into the equation (\ref{tranbe3}) yield the following TLB
evolution equations for the charges:
\begin{equation}
L\partial_{t}^{2}Q-\partial_{z}\left[  C^{-1}\partial_{z}\right]  \left(
Q+q\right)  =0, \label{tranbe7}%
\end{equation}%
\begin{equation}
\xi\left(  \partial_{t}+u_{0}\partial_{z}\right)  ^{2}q-\partial_{z}\left[
C^{-1}\partial_{z}\right]  \left(  Q+q\right)  =0,\quad\xi=\frac{4\pi}%
{\omega_{\mathrm{p}}^{2}\sigma}=\frac{m}{\sigma e\rho_{0}}>0. \label{tranbe8}%
\end{equation}

We observe now that the above evolution equations are the Euler-Lagrange
equations for the following Lagrangian%
\begin{equation}
\mathcal{L(}z,\partial_{t}Q,\partial_{z}Q,\partial_{t}q,\partial
_{z}q\mathcal{)}=\frac{L}{2}\left(  \partial_{t}Q\right)  ^{2}-\frac{1}%
{2}C^{-1}\left(  \partial_{z}Q+\partial_{z}q\right)  ^{2}+\frac{\xi}{2}\left(
\partial_{t}q+u_{0}\partial_{z}q\right)  ^{2}. \label{tranbe9}%
\end{equation}
Indeed, for a general\ Lagrangian density $\mathcal{L}=\mathcal{L}\left(
t,z;Q,\partial_{t}Q,\partial_{z}Q;q,\partial_{t}q,\partial_{z}q\right)  ,$ the
Euler-Lagrange equations take the form%
\begin{equation}
\partial_{t}\frac{\partial\mathcal{L}}{\partial(\partial_{t}Q)}+\partial
_{z}\frac{\partial\mathcal{L}}{\partial\left(  \partial_{z}Q\right)  }%
-\frac{\partial\mathcal{L}}{\partial Q}=0,\text{ \ \ \ \ }\partial_{t}%
\frac{\partial\mathcal{L}}{\partial(\partial_{t}q)}+\partial_{z}\frac
{\partial\mathcal{L}}{\partial\left(  \partial_{z}q\right)  }-\frac
{\partial\mathcal{L}}{\partial q}=0\text{\ .} \label{tranbe10}%
\end{equation}
A straightforward computation confirms that the application of equation
(\ref{tranbe10}) to the Lagrangian defined by (\ref{tranbe9}) indeed yields
the TLB evolution equations (\ref{tranbe7}) and (\ref{tranbe8}). Pierce's
original equations are obtained as a particular case, when $C,$ $L$ are
constant along the line.

As to the units of $\mathcal{L}$, they are energy/length, as expected for a
Lagrangian density. In this respect, we remind that we are using Gaussian
units and charge$^{2}=$ force$\times$length$^{2}$, in agreement with the
Gaussian version of Coulomb's law, $F=q_{1}q_{2}/r^{2}.$

Let us make a final observation: \ It is assumed that the current induced by
the beam onto the TL is due to the fact that the charge on the beam perfectly
"mirrors" onto the waveguide. This assumption can be justified as an
approximation in the "quasistatic" regime, in the spirit of Ramo's Theorem,
\cite{Ra}, \cite{Tsimring}. According to some authors, e.g. R. Kompfner,
\cite{Kom} or J.H. Booske, \cite[4]{Nusinovich}, in dealing with real devices
a coefficient $\varkappa\in(0,1)$ must be included in front of $\partial
_{z}I_{b}$ in (\ref{traneq1}) (accordingly in (\ref{tranbe2})) to account for
the real induced current, the case $\varkappa=1$ being regarded as ideal. The
Lagrangian approach can easily handle the general case. However, in order to
keep the exposition as simple as possible, we only consider the ideal situation.

\subsection{Generalization to multiple transmission lines}

It is known that fairly general wave guides can be well approximated by
multiple transmission lines, MTL, \cite{Paul}. The corresponding
generalization of Pierce's model is straightforward thanks to our Lagrangian
formulation. Indeed, suppose that we have $n+1$ conductors, one of them being
grounded, say the $(n+1)$-th. \ We denote by $V(z,t)=\left\{  V_{i}%
(z,t)\right\}  _{i=1\ldots n}$ the $n$-dimensional vector-column of voltages
\ on \ the first $n$ conductors with respect to the ground and \ by
$I(z,t)=\left\{  I_{i}(z,t)\right\}  _{i=1\ldots n}$ the vector-column of
currents flowing on them and set%
\[
Q(z,t)=\left\{  Q_{i}(z,t)\right\}  _{i=1\ldots n},\text{ \qquad}Q_{i}(z,t)=%
%TCIMACRO{\dint \limits^{t}}%
%BeginExpansion
{\displaystyle\int\limits^{t}}
%EndExpansion
I_{i}(z,s)ds.
\]
Let $L=L(z),$ $C=C(z)$ be the $n\times n$ matrices of self- and mutual
inductance and capacity. As it is well known, they are positive symmetric
(Hermitian). A natural generalization of (\ref{tranbe9}) is provided by%
\begin{equation}
\mathcal{L}=\frac{1}{2}\left\{  \left(  \partial_{t}Q,L\partial_{t}Q\right)
-\left(  \partial_{z}Q+\partial_{z}qB,C^{-1}\left[  \partial_{z}Q+\partial
_{z}qB\right]  \right)  \right\}  +\frac{\xi}{2}\left(  \partial_{t}%
q+u_{0}\partial_{z}q\right)  ^{2}, \label{mtraneq1}%
\end{equation}
where $($ $,$ $)$ stands for the scalar product in $\Re^{n}$ and $B$ is the
$n$-dimensional vector-column with all components being the unity, that is
\begin{equation}
B=\left(  1,1,\ldots1\right)  ^{\mathrm{T}}. \label{mtraneq1a}%
\end{equation}
The corresponding Euler-Lagrange second order system is%
\begin{gather}
L\partial_{t}^{2}Q-\partial_{z}\left[  C^{-1}(\partial_{z}Q+\partial
_{z}qB)\right]  =0;\label{MTLEuler-Larange}\\
\xi\left[  \partial_{t}^{2}q+2u_{0}\partial_{t}\partial_{z}q+u_{0}^{2}%
\partial_{z}^{2}q\right]  -\left(  B^{\mathrm{T}},\partial_{z}\left[
C^{-1}(\partial_{z}Q+\partial_{z}qB)\right]  \right)  =0.\nonumber
\end{gather}
The generalized telegraph equations, equivalent to the first equation above,
adopt the form%
\begin{equation}
\partial_{z}I=-C\partial_{t}V-\partial_{z}I_{b}B;\qquad\partial_{z}%
V=-L\partial_{t}I. \label{TelegMulti}%
\end{equation}
Our choice of the vector $B$ assumes, besides perfect induction, a symmetry in
the interaction between the beam and the different lines. A more realistic
approach might include coefficients $\varkappa_{i}\in(0,1)$ in vector $B$ to
account for non-symmetric interaction. As we already mentioned in Subsection
\ref{SectLagrangian}, such effects can be easily handled by our approach.

Observe that if we remove the beam\ from the system by setting $q=0,$ our
model \ is in full agreement with well established models for the interaction
of several lines, derived from Maxwell's equations under reasonable
assumptions. See, for example, \cite[2]{Nitsch}, \cite[1.4.1]{Paul} for models
of interacting TLs.

To summarize: from now on, by MTLB system we mean the field Lagrangian system
governed by the Lagrangian $\mathcal{L}$ in (\ref{mtraneq1}) and the
corresponding Euler-Lagrange field equations (\ref{MTLEuler-Larange}).

\section{The beam as a source of amplification. The role of instability
\label{SectAmplificationGeneral copy(1)}}

Evidently, the beam is the sole source of energy in the MTLB system and the
ultimate responsible for the presence of exponentially growing modes. In this
section we identify and analyze the mathematical mechanism underlying amplification.

To trace the amplification to the beam we view the Lagrangian (\ref{mtraneq1})
as a perturbation of the Lagrangian $\mathcal{L}_{\mathrm{b}}$ for the
isolated beam defined by%
\begin{equation}
\mathcal{L}_{\mathrm{b}}=\frac{1}{2}\left(  \partial_{t}q+u_{0}\partial
_{z}q\right)  ^{2}=\frac{1}{2}\left[  \left(  \partial_{t}q\right)
^{2}+2u_{0}\partial_{t}q\partial_{z}q+u_{0}^{2}\left(  \partial_{z}q\right)
^{2}\right]  . \label{Lagrbeam}%
\end{equation}
We introduce the equivalent Lagrangian $\widetilde{\mathcal{L}}=\frac{1}{\xi
}\mathcal{L}$, where $\mathcal{L}$ is as in (\ref{mtraneq1}),\textit{ i.e.}%
\begin{gather}
\widetilde{\mathcal{L}}=\mathcal{L}_{b}+\varepsilon\mathcal{L}^{\prime}%
=\frac{1}{2}\left(  \partial_{t}q+u_{0}\partial_{z}q\right)  ^{2}%
+\label{FullLagrangian}\\
+\frac{\varepsilon}{2}\left\{  \left(  \partial_{t}Q,L\partial_{t}Q\right)
-\left(  \partial_{z}Q+\partial_{z}qB,C^{-1}\left[  \partial_{z}Q+\partial
_{z}qB\right]  \right)  \right\}  ,\nonumber
\end{gather}
and $\varepsilon=1/\xi$.\ Small values of $\xi$ defined by (\ref{tranbe8}) and
consequently large values $\varepsilon$ correspond to strong coupling and
regimes where the beam effectively feeds its energy into transmission lines in
the form of EM field. The EM field energy gain originates in the beam as an
infinite reservoir of the potential energy $-\frac{1}{2}(u_{0}\partial
_{z}q)^{2}$. Importantly, the potential energy is negative unlike in
oscillatory systems. For small coupling as we will show no energy transfer
might occur from the beam to the EM field. This perturbation analysis suggests
to consider first the beam as an isolated system.

\subsection{Charge wave dynamics\label{ChgWave}}

In this subsection, we investigate beam charge dynamics as an isolated system,
described by (\ref{Lagrbeam}). We already mentioned the role of the term
$u_{0}^{2}\left(  \partial_{z}q\right)  ^{2}$ as a source of energy. This term
is responsible for the system instability manifesting itself by exponentially
growing solutions of the associated E-L equations. The gyrotropic term
$u_{0}\partial_{t}q\partial_{z}q$ in the Lagrangian provides for stabilizing
effect. As we will see, for the Lagrangian (\ref{Lagrbeam}) the balance
between instability and stability is struck exactly in the margin. Namely, a
small perturbation of this Lagrangian can make the system either stable or unstable.

The beam Lagrangian $\mathcal{L}_{\mathrm{b}}$ is quadratic in $(\partial
_{t}q,\partial_{z}q)$, see \ Section \ref{AppQuadLag}, and has the following
structure%
\begin{equation}
\mathcal{L}_{\mathrm{b}}=\frac{1}{2}\alpha(\partial_{t}q)^{2}+\theta
\partial_{t}q\partial_{z}q-\frac{1}{2}\eta(\partial_{z}q)^{2}=(\partial
_{t}q,\partial_{z}q)^{\mathrm{T}}M(\partial_{t}q,\partial_{z}q),
\label{beamlag1}%
\end{equation}
where%
\begin{equation}
M=\left[
\begin{array}
[c]{ll}%
\alpha & \theta\\
\theta & -\eta
\end{array}
\right]  =\left[
\begin{array}
[c]{ll}%
1 & u_{0}\\
u_{0} & u_{0}^{2}%
\end{array}
\right]  ,\text{ \ \ (thus }\alpha=1,\ \theta=u_{0},\ \eta=-u_{0}^{2}).
\label{beamlag2}%
\end{equation}
The corresponding Euler-Lagrange equation (\ref{qulaq4}) is
\begin{equation}
\left(  \partial_{t}+u_{0}\partial_{z}\right)  ^{2}q=0. \label{BeamEq}%
\end{equation}
Applying the general formulas (\ref{encoq4a})-(\ref{encoq4b}) for the energy
$H$ and its flux $S$ we obtain%
\begin{equation}
H_{\mathrm{b}}\left[  q\right]  =\frac{1}{2}\left(  \partial_{t}q\right)
^{2}-\frac{u_{0}^{2}}{2}\left(  \partial_{z}q\right)  ^{2}, \label{BeamEnergy}%
\end{equation}%
\begin{equation}
S_{\mathrm{b}}\left[  q\right]  =\partial_{t}q\left(  u_{0}\partial_{t}%
q+u_{0}^{2}\partial_{z}q\right)  =u_{0}\partial_{t}q\left(  \partial
_{t}q+u_{0}\partial_{z}q\right)  =u_{0}\left(  \partial_{t}q\right)
^{2}+u_{0}^{2}\partial_{t}q\partial_{z}q. \label{BeamFlux}%
\end{equation}
Since $\mathcal{L}_{\mathrm{b}}$ does not depend on time explicitly,
conservation of energy takes place, (\ref{gblag5}):%
\begin{equation}
\frac{\partial H_{\mathrm{b}}}{\partial t}+\frac{\partial S_{\mathrm{b}}%
}{\partial z}=0. \label{consenergy}%
\end{equation}

\subsubsection{Eigenmodes and stability issues.}

Since the beam parameters are constant in space we can make use of the
dispersion relation to study the eigenmodes. Thus, if we try solutions of the
form $q(z,t)=\mathrm{e}^{-\mathrm{i}(\omega t-kz)}$ in (\ref{BeamEq}), we get%
\begin{equation}
\omega^{2}-2u_{0}\omega k+u_{0}^{2}k^{2}=\left(  \omega-u_{0}k\right)  ^{2}=0,
\label{disprelbeam}%
\end{equation}
hence $k_{\omega}=\omega/u_{0}$ is a double real root. The corresponding
eigenmodes are $q_{1}(z,t)=\mathrm{e}^{\mathrm{i}\left(  k_{\omega}z-\omega
t\right)  }$\ and $q_{2}(z,t)=z\mathrm{e}^{\mathrm{i}\left(  k_{\omega
}z-\omega t\right)  }$ or their real valued counterparts
\[
v_{1}(z,t)=\cos\left(  kz-\omega t\right)  ,\qquad v_{2}(z,t)=z\cos\left(
kz-\omega t\right)  .
\]
The associated energy flux is, according to (\ref{BeamFlux}),%
\begin{equation}
S_{\mathrm{b}}\left[  v_{1}\right]  =0,\qquad S_{\mathrm{b}}\left[
v_{2}\right]  =-u_{0}^{2}z\omega\sin\left(  kz-\omega t\right)  \cos\left(
kz-\omega t\right)  . \label{fluxformulas}%
\end{equation}
To make useful inference related to conservation laws it is common to use the
following\ time-averaging operation. Namely, for a (locally integrable)
function $f$ defined on $[0,\infty)$ we introduce%
\begin{equation}
\left\langle f\right\rangle =\lim_{T\rightarrow\infty}\frac{1}{T}%
%TCIMACRO{\dint _{0}^{T}}%
%BeginExpansion
{\displaystyle\int_{0}^{T}}
%EndExpansion
f(t)\,\mathrm{d}t. \label{encoq5}%
\end{equation}
This time-averaging operation has the following properties. If $f$ is a smooth
and bounded function on $[0,\infty)$, then%
\begin{equation}
\left\langle \frac{df}{dt}\right\rangle =\lim_{T\rightarrow\infty}\frac{1}{T}%
%TCIMACRO{\dint _{0}^{T}}%
%BeginExpansion
{\displaystyle\int_{0}^{T}}
%EndExpansion
\frac{df}{dt}\,\mathrm{d}t=\lim_{T\rightarrow\infty}\frac{1}{T}\left[
f(T)-f(0)\right]  =0. \label{prop1ave}%
\end{equation}
Differentiation with respect to parameters commutes with the time-averaging
operation. Namely, if $f$ also depends (smoothly) on some parameter $z$ the
following identity holds%
\begin{equation}
\left\langle \partial_{z}f\right\rangle =\partial_{z}\left\langle
f\right\rangle . \label{prop2ave}%
\end{equation}
Taking time average on both sides \ of the conservation law (\ref{consenergy})
and using the above properties of averaging, we conclude that
\[
\left\langle S_{\mathrm{b}}\left[  v_{2}\right]  \right\rangle
(z)=\mathrm{const}.
\]
On the other hand, it \ follows from (\ref{fluxformulas}) that $\left\langle
S_{\mathrm{b}}\left[  v_{2}\right]  \right\rangle (0)=0$. Hence $\left\langle
S_{\mathrm{b}}\left[  v_{2}\right]  \right\rangle (z)=0.$

From the stability point of view, this situation is a very degenerate one. To
illustrate this point, let us introduce a special form perturbation in the
beam dispersion relation (\ref{disprelbeam}):%
\[
\omega^{2}-2\alpha u_{0}\omega k+u_{0}^{2}k^{2}=0\text{\quad or\quad}\left(
\omega-\alpha u_{0}k\right)  ^{2}=\left(  \alpha^{2}-1\right)  u_{0}^{2}%
k^{2},
\]
where $\alpha$ is a real number. Our situation corresponds to $\alpha=1.$ Let
us consider the behavior close to $\alpha=1$. The quadratic equation above has
the following solutions%
\[
k_{\omega}(\alpha)=\frac{\omega}{u_{0}}(\alpha\pm\sqrt{\alpha^{2}-1}).
\]
Notice that if $\alpha^{2}<1$ the above solutions become complex conjugate,
whereas if $\alpha^{2}>1$ they are real distinct. $\alpha=1$ \ corresponds to
a double real solution, already showing the degeneracy.

An important subject of our interest is the analysis of MTL structures in
which the parameters vary periodically in $z$. The Floquet theory and, in
particular, the Floquet multipliers are the mathematical objects that deal
with such situations \textit{par excellence}. As explained above, we may
consider the coupled system as a perturbation of the beam. Consequently, it is
instructive to take a look at the isolated beam in the light of Floquet theory
with arbitrary period (eventually dictated by the period of the structure).

The Floquet multipliers with period unity are $\rho_{\omega}(\alpha
)=\mathrm{e}^{\mathrm{i}k_{\omega}(\alpha)}$. It is clear that in the case
$\alpha^{2}<1$ they are symmetrically located with respect to the unit circle
in the complex plane. The solution corresponding to the multiplier outside the
circle is a growing wave, whereas the one corresponding to the multiplier
inside the circle is an evanescent one. In the opposite case $\alpha^{2}>1$,
both roots are located on the unit circle, and the corresponding modes are
purely oscillatory. We refer to these two qualitatively different
perturbations as respectively unstable and stable. Aiming at amplification by
coupling the beam to a MTL, the unstable situation is the one to be favoured.

The special perturbation of the beam equation considered above was for
illustration purposes to see the degenerate stability properties of the system
under perturbation of its parameters. For the MTLB system, however, it is the
MTL the one that plays the role of perturbation.\ In Section
\ref{AmplMTL-beam}, we prove that the desired instability and resulting
amplification for spatially homogeneous MTL is achieved by sufficiently strong
coupling (small values of $\xi$). An extension of this result to the case of
periodic MTLs is left for a forthcoming publication.

Additional insight into the mathematical mechanism of amplification associated
with instability can be gained by looking at the nature of the partial
differential equations involved. Indeed, the equation%
\begin{equation}
\left(  \partial_{t}^{2}+2\alpha u_{0}\partial_{t}\partial_{z}+u_{0}%
^{2}\partial_{z}^{2}\right)  q=0
\end{equation}
is\textbf{\ }hyperbolic if $\alpha^{2}>1$. In this case, there are two
propagation velocities $v^{\pm}(\alpha)$ of the same sign, namely%
\[
v^{\pm}(\alpha)=\frac{\omega}{k_{\omega}^{\pm}(\alpha)}=-u_{0}(\alpha\mp
\sqrt{\alpha^{2}-1}),
\]
and the general solution has the form%
\[
q(z,t)=q_{1}(z-v^{+}t)+q_{2}(z-v^{-}t).
\]
Therefore, any solution which is bounded in time (as it is the case for
harmonic in time solutions) is automatically bounded in space. In other words,
no harmonic in time regime can be exponentially growing in space.

\ In the critical case, $\alpha=1,$ the equation is of the parabolic type.
Changing variables $(z,t)\rightarrow(\xi,\eta)$ with $\xi=x-u_{0}t$,
$\eta=ax+bt$ with $b+au_{0}\neq0,$ it can be easily checked that the general
solution in this case is%
\[
q(z,t)=zF(z-u_{0}t)+G(z-u_{0}t)=t\widetilde{F}(z-u_{0}t)+\widetilde{G}%
(z-u_{0}t),
\]
where $F,G,\widetilde{F},\widetilde{G}$ \ are arbitrary functions. In
particular any travelling wave with velocity $u_{0}$ is a solution. Again here
we see that bounded in time dependence can be accompanied by at most linear
growth in space.

If $\alpha^{2}<1$ we are dealing with the elliptic case where there is no
propagation. This is the only case allowing for exponential amplification.
Indeed, a linear change of variables $(z,t)\rightarrow(\xi=az+bt,\eta=cz+dt)$
transforms the equation into the Laplace equation%
\[
u_{\xi\xi}+u_{\eta\eta}=0,
\]
which admits real solutions of the form $u(\xi,\eta)=\mathrm{e}^{k\xi}%
\cos(\omega\eta)$, $u(\xi,\eta)=e^{k\xi}\sin(\omega\eta),$ etc.

\section{Hamiltonian structure of the MTLB system\label{SectHamiltonian}}

In order to study the MTLB system, in particular the associated modes, their
stability and the amplification phenomenon, we make use of the Hamiltonian
structure associated to the Lagrangian (\ref{FullLagrangian}). More precisely,
we use a version of Hamiltonian formalism that treats the space and time
variables on the same footing, known as de Donder-Weyl formalism. For reader's
\ convenience, we have gathered the basic information about this topic in
Section \ref{AppdeDonder}. As usual, \ the passage from Lagrangian to
Hamiltonian point of view allows to cast the second-order Euler-Lagrange
system of equations (\ref{MTLEuler-Larange}) in the form of a first order
system, either with respect to $t$ or with respect to $z.$

To comply with notations of Section \ref{AppdeDonder}, from now on we put
$\mathsf{q}_{1}=Q$, $\mathsf{q}_{2}=q$, $\mathsf{q}=(\mathsf{q}_{1}%
,\mathsf{q}_{2})^{^{\mathrm{T}}}$. The Lagrangian$\ \widetilde{\mathcal{L}}$
in (\ref{FullLagrangian}) is quadratic in its variables $\left(  \partial
_{t}Q,\partial_{t}q,\partial_{z}Q,\partial_{z}q\right)  $, that is in $\left(
\partial_{t}\mathsf{q},\partial_{z}\mathsf{q}\right)  ^{^{\mathrm{T}}}$ in the
new notation. Indeed,%
\begin{equation}
\widetilde{\mathcal{L}}=\frac{1}{2}\partial_{t}\mathsf{q}^{\mathrm{T}}%
\alpha\partial_{t}\mathsf{q}+\partial_{t}\mathsf{q}^{\mathrm{T}}\theta
\partial_{z}\mathsf{q}-\frac{1}{2}\partial_{z}\mathsf{q}^{\mathrm{T}}%
\eta\partial_{z}\mathsf{q}, \label{Lagqq1}%
\end{equation}
where%
\begin{equation}
\alpha=\left[
\begin{array}
[c]{cc}%
\varepsilon L & 0\\
0 & 1
\end{array}
\right]  ,\text{\quad}\theta=\left[
\begin{array}
[c]{cc}%
0 & 0\\
0 & u_{0}%
\end{array}
\right]  ,\text{\quad}\eta=\left[
\begin{array}
[c]{cc}%
\varepsilon C^{-1} & \varepsilon C^{-1}B\\
\varepsilon B^{\mathrm{T}}C^{-1} & \varepsilon B^{\mathrm{T}}C^{-1}B-u_{0}^{2}%
\end{array}
\right]  \text{,} \label{Lagqq2}%
\end{equation}
or, using a block matrix,%
\begin{equation}
\widetilde{\mathcal{L}}=\frac{1}{2}\mathsf{u}^{\mathrm{T}}M_{\mathrm{L}%
}\mathsf{u},\text{\quad with}\qquad M_{\mathrm{L}}=\left[
\begin{array}
[c]{ll}%
\alpha & \theta\\
\theta & -\eta
\end{array}
\right]  ,\qquad\mathsf{u}=\left[
\begin{array}
[c]{l}%
\partial_{t}\mathsf{q}\\
\partial_{z}\mathsf{q}%
\end{array}
\right]  . \label{Lagqq3}%
\end{equation}
Let us introduce the vector of canonical momenta $\mathsf{p}=(\mathsf{p}%
_{t},\mathsf{p}_{z})^{^{\mathrm{T}}},$ related to the vector $\mathsf{u}$
above by means of%
\begin{equation}
\mathsf{p}=M_{\mathrm{L}}\mathsf{u,} \label{momentum}%
\end{equation}
where $M_{\mathrm{L}}$ is as in \ref{Lagqq3}. In the following result, we
express the dynamics of our system in terms of the variables $\mathsf{p}_{z}$
and $\partial_{t}\mathsf{q}$.

\begin{theorem}
The second order Euler-Lagrange system (\ref{MTLEuler-Larange}) is equivalent
to the $2n-$first order system
\begin{equation}
\tilde{J}\partial_{z}V=\mathrm{i}\partial_{t}\tilde{M}V,\qquad V=\left[
\begin{array}
[c]{l}%
\mathsf{p}_{z}\\
\partial_{t}\mathsf{q}%
\end{array}
\right]  , \label{HamEq1}%
\end{equation}
where%
\begin{equation}
\tilde{J}=\left[
\begin{array}
[c]{cc}%
0 & \mathrm{i}\mathbf{1}\\
\mathrm{i}\mathbf{1} & 0
\end{array}
\right]  ,\qquad\tilde{M}=\tilde{M}\left(  z\right)  =\left[
\begin{array}
[c]{ll}%
-\eta\left(  z\right)  ^{-1} & \eta\left(  z\right)  ^{-1}\theta\\
\theta\eta\left(  z\right)  ^{-1} & -\alpha\left(  z\right)  -\theta
\eta\left(  z\right)  ^{-1}\theta
\end{array}
\right]  . \label{aeta4}%
\end{equation}

\end{theorem}

\begin{proof}
The derivation of de Donder-Weyl version of Hamilton equations in the variable
$z$ for general quadratic Lagrangians is described in Section \ref{AppQuadLag}%
. In particular, for $\widetilde{\mathcal{L}}$ \ defined by (\ref{Lagqq1}) the
Hamiltonian $H_{\mathrm{DW}}\left(  \mathsf{p}\right)  $ does not depend
explicitly on $\mathsf{q}$ and%
\[
H_{\mathrm{DW}}\left(  \mathsf{p}\right)  =\widetilde{\mathcal{L}}\left(
\mathsf{u}\right)  ,
\]
where $\mathsf{p}$ is linked to $\mathsf{u}$ as in (\ref{momentum}).
Equivalently,%
\begin{equation}
H_{\mathrm{DW}}(\mathsf{p})=\frac{1}{2}\mathsf{p}^{\mathrm{T}}M_{\mathrm{L}%
}^{-1}\mathsf{p}. \label{Lagqq4}%
\end{equation}
According to (\ref{mhzne2}), (\ref{mhzne3}), in the variables $(\mathsf{p}%
_{z},\partial_{t}\mathsf{q})^{^{\mathrm{T}}}$, the corresponding first-order
system is precisely (\ref{HamEq1}), \ with $\ \tilde{J}$ \ and $\ \tilde{M}$
\ as in (\ref{aeta4}).
\end{proof}

We recall that $\tilde{J}$ \ and $\ \tilde{M}$ are, respectively,
antihermitian and hermitian,\textit{ i.e.}%
\[
\tilde{J}^{\ast}=-\tilde{J},\qquad\tilde{M}^{\ast}=\tilde{M}.
\]

Consider now\ a time harmonic solution of the form $\mathsf{q(}%
z,t)=\widehat{\mathsf{q}}\left(  z\right)  \mathrm{e}^{-\mathrm{i}\omega t}.$
In this case,
\begin{equation}
V(z,t)=\left[
\begin{array}
[c]{l}%
\mathsf{p}_{z}\\
\partial_{t}\mathsf{q}%
\end{array}
\right]  =\hat{V}\left(  z\right)  \mathrm{e}^{-\mathrm{i}\omega t},\qquad
\hat{V}(z)=\left[
\begin{array}
[c]{l}%
\widehat{\mathsf{p}}_{z}\left(  z\right) \\
-\mathrm{i}\omega\widehat{\mathsf{q}}\left(  z\right)
\end{array}
\right]  \label{aeta4a}%
\end{equation}
and the Hamiltonian equation (\ref{HamEq1}) is reduced to%
\begin{equation}
\tilde{J}\partial_{z}\hat{V}=\omega\tilde{M}\hat{V}. \label{aeta3}%
\end{equation}
Notice that the equation (\ref{aeta3}) for $\hat{V}$ is Hamiltonian according
to the definition in Section \ref{AppCanHam}, and the conservation law
(\ref{Jzhz10}) applies, yielding
\begin{gather}
\hat{V}^{\ast}\tilde{J}\hat{V}=\mathrm{i}\left[  \widehat{\mathsf{p}}%
_{z}^{\ast}\left(  -\mathrm{i}\omega\widehat{\mathsf{q}}\left(  z\right)
\right)  +\left(  -\mathrm{i}\omega\widehat{\mathsf{q}}\left(  z\right)
\right)  ^{\ast}\widehat{\mathsf{p}}_{z}\right]  =2\mathrm{i}%
\operatorname*{Re}\left\{  \left(  -\mathrm{i\omega}\widehat{\mathsf{q}%
}\left(  z\right)  \right)  ^{\ast}\widehat{\mathsf{p}}_{z}\right\}
\label{aeta5}\\
=-2\mathrm{i}\omega\operatorname{Im}\left\{  \left(  \widehat{\mathsf{q}%
}\left(  z\right)  \right)  ^{\ast}\left[  \theta\left(  -\mathrm{i}%
\omega\right)  \widehat{\mathsf{q}}\left(  z\right)  -\eta\left(  z\right)
\partial_{z}\widehat{\mathsf{q}}\left(  z\right)  \right]  \right\}
=\operatorname*{constant}.\nonumber
\end{gather}
Later on, in Section \ref{SubsEnergyExchange}, we will see how the above
conservation law relates to energy flux constancy.

\section{Amplification for the homogeneous case\label{AmplMTL-beam}}

This subsection is devoted to the analysis of the amplification regime
associated with a single exponentially growing mode in the case of an
homogeneous MTLB system, that is with parameters not varying with $z$. For
real $\omega$ we seek solutions of (\ref{MTLEuler-Larange}) in the form%
\begin{equation}
Q(z,t)=\widehat{Q}\mathrm{e}^{-\mathrm{i}(\omega t-kz)},\qquad
q(z,t)=\widehat{q}\mathrm{e}^{-\mathrm{i}(\omega t-kz)},
\label{PlaneWavesStructure}%
\end{equation}
where $\widehat{q}$ and $k$ are complex constants and $\widehat{Q}$ is a
complex vector. We show that, under certain conditions, there is a solution
with genuinely complex, that is, non real wave number $k$.

Let us recall that the eigenvelocities of the MTL are the roots of the
equation%
\[
\left\vert C^{-1}-v^{2}L\right\vert =0,
\]
\cite{Paul}, \cite{Nitsch}. \ Since both $L$ and $C$ are positive definite,
the symmetric $n\times n$ matrix $L^{-1/2}C^{-1}L^{-1/2}$ has positive
eigenvalues $0<\lambda_{1}\leq\lambda_{2}\leq...\leq\lambda_{n}$, where
multiple eigenvalues are repeated according to their multiplicity. Then the
MTL has characteristic velocities are precisely%
\begin{equation}
\pm v_{i},\text{ where }v_{i}^{2}=\lambda_{i}. \label{chav1}%
\end{equation}

\begin{theorem}
\label{TeoremAmplification}Let $u_{0},\xi>0.$ If either
\end{theorem}

\begin{enumerate}
\item[(i)] $0<u_{0}\leq v_{1}$ \textit{or}

\item[(ii)] $v_{1}<u_{0}$ \textit{and} $\xi>0$ \textit{is sufficiently small,}

\noindent\textit{then for each real }$\omega$\textit{ there are exactly two
genuinely complex conjugate values }$k_{0}$\textit{ and }$k_{0}^{\ast}%
$\textit{ such that (\ref{PlaneWavesStructure}) is a non-trivial solution of
equations (\ref{MTLEuler-Larange}).}
\end{enumerate}

Hence, assuming $\operatorname*{Im}k_{0}<0$ we have the associated solution
\begin{equation}
Q(z,t)=A(z)\mathrm{e}^{-\mathrm{i}\omega t}\mathrm{e}^{-(\operatorname*{Im}%
k_{0})z},\qquad q(z,t)=B(z)\mathrm{e}^{-\mathrm{i}\omega t}\mathrm{e}%
^{-(\operatorname*{Im}k_{0})z},\qquad A(z),B(z)\neq0, \label{chav2}%
\end{equation}
that grows exponentially in the $+z$ direction, whereas the solution
associated with $k_{0}^{\ast}$ decays exponentially.

We sketch the proof, deferring the mathematical details to section
\ref{AppAmplification}. Substituting the expressions
(\ref{PlaneWavesStructure}) into the system (\ref{MTLEuler-Larange}), we
obtain the following linear algebraic system of $n+1$ equations for
$\widehat{Q},\widehat{q}$:%
\begin{equation}
\left[
\begin{array}
[c]{cc}%
-v^{2}L+C^{-1} & D\\
D^{T} & d-\xi(v-u_{0})^{2}%
\end{array}
\right]  \left[
\begin{array}
[c]{c}%
\widehat{Q}\\
\widehat{q}%
\end{array}
\right]  =\left[
\begin{array}
[c]{c}%
0\\
0
\end{array}
\right]  ,\quad\text{where }v=\frac{\omega}{k} \label{Systemforv}%
\end{equation}
and%
\begin{equation}
D=(D_{i}),\qquad D_{i}=%
%TCIMACRO{\dsum \limits_{j}}%
%BeginExpansion
{\displaystyle\sum\limits_{j}}
%EndExpansion
(C^{-1})_{ij},\qquad d=%
%TCIMACRO{\dsum \limits_{i}}%
%BeginExpansion
{\displaystyle\sum\limits_{i}}
%EndExpansion
D_{i}. \label{ddcj1}%
\end{equation}
For the sake of brevity, we denote%
\begin{equation}
A(v)=-v^{2}L+C^{-1},\qquad\widetilde{A}(v)=\left[
\begin{array}
[c]{cc}%
A(v) & D\\
D^{T} & d-\xi(v-u_{0})^{2}%
\end{array}
\right]  . \label{ddcj2}%
\end{equation}
The system (\ref{Systemforv}) has nontrivial solutions if and only if
$\ \left\vert \widetilde{A}(v)\right\vert =0$. The corresponding polynomial
equation of degree $2n+2$ is the dispersion relation of our system written in
terms of the velocity $v$. In Section \ref{AppAmplification} we prove in full
detail that the equation $\left\vert \widetilde{A}(v)\right\vert =0$ has
exactly one pair of complex conjugate solutions if either (i) or (ii) holds.
Here we outline the main ideas of the proof.

First of all, if $\left\vert A(v)\right\vert \neq0$, the following
\emph{canonical factorization }takes place%
\begin{equation}
\left\vert \widetilde{A}(v)\right\vert =\left\vert A(v)\right\vert \left[
d-\xi(v-u_{0})^{2}-D^{T}(A(v))^{-1}D\right]  . \label{CanFact}%
\end{equation}
The values of $v$ such that $\left\vert A(v)\right\vert =0$ are precisely the
eigenvelocities $\pm v_{i}$ of the waveguide. Therefore, the roots of
$\left\vert \widetilde{A}(v)\right\vert =0$ different from $\pm v_{i}$,
$i=1,2,...n$ \ are the roots of the equation%
\begin{equation}
-\xi(v-u_{0})^{2}=R(v),\text{ where }R(v)=D^{T}(A(v))^{-1}D-d,
\label{DispRelGen}%
\end{equation}
in which the two components of the system enter separately. The rational
function $R(v)$ in (\ref{DispRelGen}) contains the relevant information about
the MTL whereas the left hand side depends only on the beam parameters. In
what follows we refer to function $R(v)$ as \emph{MTL characteristic
function}. It can be explicitly written in terms of the characteristic
velocities:%
\begin{equation}
R(v)=%
%TCIMACRO{\dsum \limits_{i=1}^{n}}%
%BeginExpansion
{\displaystyle\sum\limits_{i=1}^{n}}
%EndExpansion
\frac{\widetilde{D}_{i}^{2}}{v_{i}^{2}-v^{2}}-d, \label{R}%
\end{equation}
where $\widetilde{D}_{i}$ are constants related to $D_{i}$, see Section
\ref{AppAmplification}. The graph of the MTL characteristic function $R$ is
symmetric with respect to the vertical axis and is made up of branches, \ a
central one with the minimum at $(0,0)$, a number of increasing \ branches for
$v>0$ and decreasing for $v<0$. One can readily see that $\lim_{v\rightarrow
\infty}R(v)=-d$. In addition to that, the graph of $R$ has vertical asymptotes
at $v=\pm v_{i}$ if at least one of the associated $\widetilde{D}_{j}$ does
not vanish. The number of the asymptotes varies between $2$ and $2n$. The
left-hand side in (\ref{DispRelGen}) is a parabola with vertex at $(u_{0},0)$.

Figure \ref{FigAmpl} shows the graph of $R$ and that of the parabola
$\ y=-\xi(v-u_{0})^{2}$ with the following inductance and capacity matrices%
\[
L=\left[
\begin{array}
[c]{ccc}%
4 & 1 & 1/2\\
1 & 5 & 2\\
1/2 & 2 & 2
\end{array}
\right]  ;\qquad C=\left[
\begin{array}
[c]{ccc}%
2 & 1 & 2\\
1 & 4 & 0\\
2 & 0 & 1
\end{array}
\right]  .
\]
The approximate \ values of the characteristic velocities are: $v_{1}=0.18357$
and $v_{2}=0.42383$. In Figure \ref{FigAmpl} (a), $u_{0}=0.18$ and $\xi=2$; in
Figure \ref{FigAmpl} (b), $u_{0}=0.8$ and $\xi=18$.

It is important to observe that the parabola always intersects all the
branches of $R$ except for the central one. For small $\xi$ each branch is
intersected only once, and consequently the number of real roots of the
equation (\ref{DispRelGen}) is exactly the number of asymptotes, as in Figure
\ref{FigAmpl} (a) above. For large $\xi$ however the number of real roots can
exceed the number of asymptotes as in Figure \ref{FigAmpl} (b), where a large
value of $\xi$ produces three points of intersection with the far right branch
of the graph of $R$. Moreover, if $\ u_{0}\leq v_{1}$ (geometrically, the
vertex of the parabola lies between the vertical axis and the first
asymptote), then clearly the number of real roots equals the number of
asymptotes irrespective of the value of $\xi>0$. These facts can be proved
rigorously based on monotonicity properties, but their geometric
interpretation is so transparent that a quick look at Figure \ref{FigAmpl} is
quite convincing.%
%TCIMACRO{\FRAME{ftbpFU}{6.4264in}{2.4794in}{0pt}{\Qcb{(a) $u_{0}<v_{1}$: the
%parabola $y=-\xi(v-u_{0})^{2}$ (dashed line) intersects each branch of
%$y=R(v)$ (where $R$ is as in (\ref{R})) just once. Four real roots. \ (b)
%$u_{0}>v_{1}:$ for large $\xi,$ the parabola intersects one of the branches of
%$y=R(v)$ three times. Six real roots.}}{\Qlb{FigAmpl}}{fig12amp.eps}%
%{\special{ language "Scientific Word";  type "GRAPHIC";
%maintain-aspect-ratio TRUE;  display "USEDEF";  valid_file "F";
%width 6.4264in;  height 2.4794in;  depth 0pt;  original-width 9.5986in;
%original-height 3.6789in;  cropleft "0";  croptop "1";  cropright "1";
%cropbottom "0";  filename 'Fig12Amp.eps';file-properties "XNPEU";}} }%
%BeginExpansion
\begin{figure}[ptb]%
\centering
\ifcase\msipdfoutput
\includegraphics[
height=2.4794in,
width=6.4264in
]%
{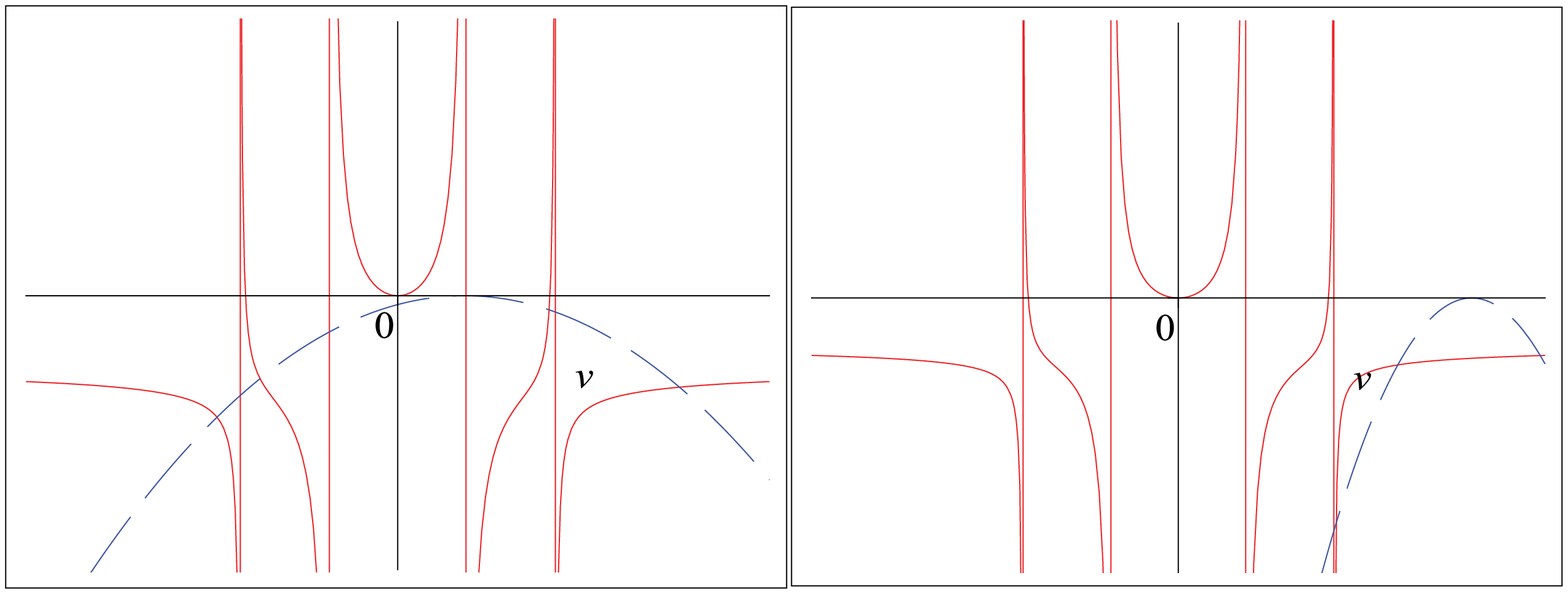}%
\else
\includegraphics[
height=2.4794in,
width=6.4264in
]%
{D:/alatex/Preparation/Reyes/lagsys/arxiv2/graphics/Fig12Amp__2.pdf}%
\fi
\caption{(a) $u_{0}<v_{1}$: the parabola $y=-\xi(v-u_{0})^{2}$ (dashed line)
intersects each branch of $y=R(v)$ (where $R$ is as in (\ref{R})) just once.
Four real roots. \ (b) $u_{0}>v_{1}:$ for large $\xi,$ the parabola intersects
one of the branches of $y=R(v)$ three times. Six real roots.}%
\label{FigAmpl}%
\end{figure}
%EndExpansion
\ In the generic case, the roots of the dispersion relation $\left\vert
\widetilde{A}(v)\right\vert =0$ are exactly those of equation
(\ref{DispRelGen}), but in general some of the $v_{i}$ can also be roots.
Whenever some $v_{i}$ is a real root (maybe multiple) of $\left\vert
\widetilde{A}(v)\right\vert =0$, the number of asymptotes in the graph of $R$
is reduced by the corresponding amount. The same is true of the number of real
roots of (\ref{DispRelGen}) under either (i) or (ii). This fact follows from
factorization (\ref{CanFact}). The main point is that in all cases the total
number of real roots of $\left\vert \widetilde{A}(v)\right\vert =0$ \ is $2n$
\ if either condition (i) or (ii) above holds. We thus conclude that under (i)
or (ii) there is necessarily a \emph{unique} pair of complex conjugate roots.
The detailed proof of these facts is provided in Section
\ref{AppAmplification}.

It is not difficult to estimate how small $\xi$ should be in condition (ii)
above. In the case $n=1$ and $u_{0}>v_{1}$ there is a precise criterion on
$\xi$, namely amplification takes place if%
\begin{equation}
\xi<\xi_{0}:=\frac{L\gamma^{2}}{1-\gamma^{2/3}};\qquad\gamma=\frac{v_{1}%
}{u_{0}}=\frac{1}{u_{0}\sqrt{LC}}. \label{xismall}%
\end{equation}
A simple sufficient condition can be also given for $n>1$. For example, one
can just impose that the left branch of the parabola at $v=0$ be flatter than
the flattest point of the graph of $R$ on $(v_{1},u_{0}).$This leads to%
\begin{equation}
\xi<\widetilde{\xi}_{0}:=\frac{\min_{v\in(v_{1},u_{0})}R^{\prime}(v)}{2u_{0}}.
\label{xismallbis}%
\end{equation}
Observe that both $\xi_{0}$ and $\widetilde{\xi}_{0}$ vanish as $u_{0}%
\rightarrow\infty$, as expected. The value of $\widetilde{\xi}_{0}$ is not
sharp but we did not make an effort to find one.

Thus, under the above assumptions the system exhibits spatially exponentially
growing, as well as exponentially decaying \ time harmonic regimes. Using the
terminology of dynamical systems, if we restrict to time harmonic evolutions
$z\rightarrow X,$ where $X=\left\{  e^{-i\omega t}\widetilde{Q},\quad
\widetilde{Q}\in%
%TCIMACRO{\U{2102} }%
%BeginExpansion
\mathbb{C}
%EndExpansion
^{n+1}\right\}  ,$ there is a subspace of data \ inducing an individual
exponential dichotomy for $X,$ both for forward and backward in $z$
evolutions, \cite[XIII]{Ha}. The subspace is determined by the solutions
$\widetilde{Q}=(\widehat{Q},\widehat{q})$ of the system (\ref{Systemforv})
with $v$ being the corresponding complex solution of $\left\vert
\widetilde{A}(v)\right\vert =0.$

\subsection{Asymptotic behavior of the amplification factor as $\xi
\rightarrow0$ and as $\xi\rightarrow\infty$.\label{BehAmplificationFactor}}

Let $k_{0}$ denote the complex root with $\operatorname*{Im}k_{0}<0$ whose
existence we proved in the previous section under appropriate conditions. It
is interesting to study the asymptotics of the "amplification factor"
$-\operatorname*{Im}k_{0}$ as the beam parameter $\xi\rightarrow0,$ as well as
its behavior when $\xi\rightarrow\infty.$ A careful analysis shows (see
Section \ref{AppAmplification}) that, if we denote by $v_{0}=\omega/k_{0}$ the
corresponding velocity with $\operatorname*{Im}v_{0}>0$, then
\[
\operatorname*{Im}v_{0}=\sqrt{K^{\prime}\xi+o(\xi)}=\sqrt{K^{\prime}}\sqrt
{\xi}+o(\sqrt{\xi})\text{ \ as }\xi\rightarrow0,
\]
where $K^{\prime}$ depends only on $L,C,u_{0}.$ As a consequence,%
\begin{equation}
-\operatorname*{Im}k_{0}=\frac{\operatorname*{Im}v_{0}}{\text{\ \ }\left\vert
v_{0}\right\vert ^{2}}\sim\frac{K^{\prime\prime}}{\sqrt{\xi}\text{\ }%
}\ \text{\ as \ }\xi\rightarrow0;\ \ K^{\prime\prime}>0. \label{AmpFactorAt0}%
\end{equation}
The conclusion is that, in this model, the amplification factor can be
indefinitely improved by reducing $\xi.$ According to (\ref{tranbe8}), this
amounts to increasing $\sigma\rho_{0},$ the linear electron density of the beam.

On the other hand, the limit $\xi\rightarrow\infty$ makes sense only if
$0<u_{0}\leq v_{1}$. In the case of one line and $u_{0}=v_{1}$, it can be
proved that%
\begin{equation}
-\operatorname*{Im}k_{0}=\frac{\operatorname*{Im}v_{0}}{\text{\ \ }\left\vert
v_{0}\right\vert ^{2}}\sim\frac{K^{\prime\prime\prime}}{\sqrt[3]{\xi}%
}\ \text{\ as \ }\xi\rightarrow\infty;\ \ K^{\prime\prime\prime}>0,
\label{AmpFactorAtInfty}%
\end{equation}
see Section \ref{AppAmplification}.

The regime considered by Pierce corresponds to the latter situation, in which
there are two real solutions (for $v$) close to $\pm u_{0,}$ and two complex
conjugate with real part close to $u_{0}$; see Section \ref{PierceRev copy(1)}%
. The situation is similar for $u_{0}<v_{1}$, but in this case
$-\operatorname*{Im}k_{0}$ has a finite positive limit as $\xi\rightarrow
\infty$.

\section{Energy conservation and transfer\label{EnergyConsEx}}

The conservation laws for our system can be obtained via Noether theorem,
\cite[38.2-3]{GelFom}, \cite[13.7]{Gold}.

\begin{theorem}
Conservation of energy \ for the system (\ref{MTLEuler-Larange}) holds in the
form%
\begin{equation}
\partial_{t}H+\partial_{z}S=0, \label{encoq2}%
\end{equation}
where the total energy $H$ and the total energy flux $S$ are given by
\begin{equation}
H=\frac{1}{2}\partial_{t}\mathsf{q}^{\mathrm{T}}\alpha\partial_{t}%
\mathsf{q}+\frac{1}{2}\partial_{z}\mathsf{q}^{\mathrm{T}}\eta\partial
_{z}\mathsf{q}; \label{encoq4a}%
\end{equation}%
\begin{equation}
S=\partial_{t}\mathsf{q}^{\mathrm{T}}\theta\partial_{t}\mathsf{q}-\partial
_{t}\mathsf{q}^{\mathrm{T}}\eta\partial_{z}\mathsf{q}=\partial_{t}%
\mathsf{q}^{\mathrm{T}}\left(  \theta\partial_{t}\mathsf{q}-\eta\partial
_{z}\mathsf{q}\right)  =\partial_{t}\mathsf{q}^{\mathrm{T}}\mathsf{p}_{z}.
\label{encoq4b}%
\end{equation}

\end{theorem}

\begin{proof}
The Lagrangian density $\mathcal{L}$ does not depend explicitly on $t$ (this
is a consequence of the closedness of the system), therefore by the fields
version of Noether theorem, \cite[38.2-3]{GelFom}, \cite[13.7]{Gold},
conservation of energy (\ref{encoq2}) holds, with energy and the energy flux
densities given by%
\begin{equation}
H=%
%TCIMACRO{\dsum _{j}}%
%BeginExpansion
{\displaystyle\sum_{j}}
%EndExpansion
\frac{\partial\mathcal{L}}{\partial(\partial_{t}\mathsf{q}_{j})}\partial
_{t}\mathsf{q}_{j}-\mathcal{L},\quad S=%
%TCIMACRO{\dsum _{j}}%
%BeginExpansion
{\displaystyle\sum_{j}}
%EndExpansion
\frac{\partial\mathcal{L}}{\partial(\partial_{z}\mathsf{q}_{j})}\partial
_{t}\mathsf{q}_{j}. \label{encoq1}%
\end{equation}
A straightforward computation yields the expressions of $H$ and $S$ given in
(\ref{encoq4a}), (\ref{encoq4b}). In (\ref{encoq4b}), $\mathsf{p}_{z}$ is the
canonical momentum defined in (\ref{qulaq5}), Section \ref{AppQuadLag}.
\end{proof}

Consider now a real time harmonic eigenmode%
\begin{equation}
\mathsf{q}\left(  t,z\right)  =\operatorname{Re}\left\{  \mathsf{\hat{q}%
}\left(  z\right)  \mathrm{e}^{-\mathrm{i}\omega t}\right\}  ,\text{ with a
complex valued }\mathsf{\hat{q}}\left(  z\right)  , \label{encoq4c}%
\end{equation}
which solves the Euler-Lagrange equation (\ref{qulaq4}). Notice that
$\left\langle \mathsf{q}\right\rangle (z)=0$, where $\left\langle
\cdot\right\rangle $ is the time average operation defined in (\ref{encoq5}).
However, if%
\begin{equation}
a\left(  t\right)  =\operatorname{Re}\left\{  \hat{a}\mathrm{e}^{-\mathrm{i}%
\omega t}\right\}  ,\text{ with a complex valued }\hat{a} \label{encoq5a}%
\end{equation}
and $b\left(  t\right)  $ is defined by a similar formula then we have%
\begin{equation}
\left\langle ab\right\rangle =\frac{1}{2}\operatorname{Re}\left\{  \hat
{a}^{\ast}\hat{b}\right\}  . \label{encoq5b}%
\end{equation}
Applying the averaging operation $\left\langle \cdot\right\rangle $ to the
conservation law (\ref{encoq2}) for a time harmonic eigenmode $q$ as in
(\ref{encoq4c}) and using (\ref{prop1ave}) and (\ref{prop2ave}), we obtain%
\begin{equation}
\partial_{z}\left\langle S\right\rangle \left(  z\right)  =0\text{ implying
}\left\langle S\right\rangle \left(  z\right)  =\operatorname*{constant}.
\label{encoq6}%
\end{equation}
On the other hand, $S$ defined by (\ref{encoq4b}) can be written as the
product of two real time harmonic functions:%
\begin{equation}
S(t,z)=\operatorname{Re}\left\{  \widehat{A}(z)\mathrm{e}^{-\mathrm{i}\omega
t}\right\}  \operatorname{Re}\left\{  \widehat{B}(z)\mathrm{e}^{-\mathrm{i}%
\omega t}\right\}  , \label{encoq6b}%
\end{equation}
where%
\begin{equation}
\widehat{A}(z)=-i\omega\mathsf{\hat{q}}\left(  z\right)  ;\qquad
\widehat{B}(z)=-i\omega\theta\mathsf{\hat{q}}(z)-\eta\partial_{z}%
\mathsf{\hat{q}}(z). \label{encoq6c}%
\end{equation}
Using (\ref{encoq5b}) we obtain the energy flux conservation law in the form
\begin{equation}
\left\langle S\right\rangle \left(  z\right)  =\frac{1}{2}\operatorname{Re}%
\left\{  \left\langle \left(  -i\omega\mathsf{\hat{q}}\right)  ^{\ast}\left(
-i\omega\theta\mathsf{\hat{q}}-\eta\partial_{z}\mathsf{\hat{q}}\right)
\right\rangle \right\}  =\frac{1}{2}\operatorname{Re}\left\{  \left\langle
\left(  -i\omega\mathsf{\hat{q}}\right)  ^{\ast}\mathsf{\hat{p}}%
_{z}\right\rangle \right\}  =\operatorname*{constant}. \label{encoq6a}%
\end{equation}
Constancy of $\left\langle S\right\rangle \left(  z\right)  $ is related to
the constancy of the symplectic square of the solution of the Hamiltonian
system satisfied by%
\begin{equation}
\widehat{V}(z)=\left[
\begin{array}
[c]{l}%
\widehat{\mathsf{p}}_{z}\\
-i\omega\widehat{\mathsf{q}}%
\end{array}
\right]  , \label{encoq7}%
\end{equation}
see formula (\ref{aeta5}). Indeed,%
\begin{equation}
V^{\ast}\tilde{J}V=2i\operatorname*{Re}\left\{  \left(  -i\omega
\mathsf{\hat{q}}\left(  z\right)  \right)  ^{\ast}\left[  -\mathrm{i}%
\omega\theta\mathsf{\hat{q}}\left(  z\right)  -\eta\partial_{z}\mathsf{\hat
{q}}\left(  z\right)  \right]  \right\}  =\operatorname*{constant}%
=4\mathrm{i}\left\langle S\right\rangle \left(  z\right)  . \label{encoq8}%
\end{equation}

\subsection{Energy exchange between subsystems\label{SubsEnergyExchange}}

This section deals with the balance of energy between the two subsystems
making up our system: the beam and the MTL. As already pointed out by Pierce
in \cite[p. 635]{Pier51} an amplification regime assumes that the energy
extracted from the beam is stored in the EM field. In other words, the net
flux of energy must have a definite sign.\ Pierce tacitly considers this
condition as an additional one to be imposed on top of other conditions
ensuring the existence of an exponentially growing solution. We show below
that in fact this condition is automatically satisfied for exponentially
growing solutions.

When computing the energy flux between the beam and the MTL we take advantage
of our Lagrangian setting. This setting allows for a systematic derivation of
expressions for energies and fluxes satisfying \textit{a priori} the
fundamental conservation laws. We proceed using the results from Section
\ref{AppEnergyExchange} for a more general coupled system.

First, we should split the Lagrangian into two parts $\mathcal{L=L}%
_{1}\mathcal{+L}_{2}$ corresponding to the MTL and the beam. Namely,%
\begin{align}
\mathcal{L}_{1}(Q_{t},Q_{;z})  &  =\frac{1}{2}\left(  \partial_{t}%
Q,L\partial_{t}Q\right)  ^{2}-\frac{1}{2}\left(  \partial_{_{;z}}%
Q,C^{-1}\partial_{_{;z}}Q\right)  ^{2};\label{LLQq1}\\
\mathcal{L}_{2}(q_{t},q_{z})  &  =\frac{\xi}{2}(\partial_{t}q+u_{0}%
\partial_{z}q)^{2},\text{ where }\partial_{;z}Q=\partial_{z}Q+B\partial
_{z}q.\nonumber
\end{align}
The above Lagrangian has the structure of (\ref{gblag1}), with
$\ B=(1,1...1)^{\mathrm{T}}$. Our first result concerning energy flows is
contained in the following

\begin{theorem}
The instantaneous power by unit length supplied by the beam to the MTL is
given by
\begin{equation}
P_{\mathrm{B}\rightarrow\mathrm{MTL}}=\partial_{t}\left[  \frac{1}{2}\left(
CV,V\right)  +\frac{1}{2}\left(  LI,I\right)  \right]  +\partial_{z}(I,V)
\label{exprpower1}%
\end{equation}
where, as usual, $(,)$ stands for the scalar product.
\end{theorem}

\begin{proof}
According to (\ref{gblag11}), the power $P_{\mathrm{B}\rightarrow\mathrm{MTL}%
}$ flowing from the beam to the (unit length of) the MTL is given by%
\begin{gather}
P_{\mathrm{B}\rightarrow\mathrm{MTL}}=-\frac{\partial\mathcal{L}_{1}}%
{\partial(\partial_{;z}Q)}\partial_{tz}^{2}q=\partial_{;z}Q^{T}C^{-1}%
B\partial_{tz}^{2}q=\partial_{z}I_{b}%
%TCIMACRO{\dsum \limits_{i}}%
%BeginExpansion
{\displaystyle\sum\limits_{i}}
%EndExpansion
D_{i}\partial_{;z}Q_{i},\label{LLQq2}\\
\text{where }D_{i}=%
%TCIMACRO{\dsum \limits_{j}}%
%BeginExpansion
{\displaystyle\sum\limits_{j}}
%EndExpansion
(C^{-1})_{ij}.\nonumber
\end{gather}
Using (\ref{TelegMulti}) we recast the expression for $P_{\mathrm{B}%
\rightarrow\mathrm{MTL}}$ in terms of currents and voltages. Indeed, the
voltage $V$ \ is given by%
\begin{equation}
V=-C^{-1}(\partial_{z}Q+\partial_{z}q). \label{LLQq3}%
\end{equation}
Then we notice that
\[%
%TCIMACRO{\dsum \limits_{i}}%
%BeginExpansion
{\displaystyle\sum\limits_{i}}
%EndExpansion
D_{i}\partial_{;z}Q_{i}=%
%TCIMACRO{\dsum \limits_{i}}%
%BeginExpansion
{\displaystyle\sum\limits_{i}}
%EndExpansion%
%TCIMACRO{\dsum \limits_{j}}%
%BeginExpansion
{\displaystyle\sum\limits_{j}}
%EndExpansion
(C^{-1})_{ij}(\partial_{z}Q_{i}+\partial_{z}qB_{i})=-%
%TCIMACRO{\dsum \limits_{j}}%
%BeginExpansion
{\displaystyle\sum\limits_{j}}
%EndExpansion
V_{j}%
\]
and hence, according to (\ref{TelegMulti}),
\begin{gather}
P_{\mathrm{B}\rightarrow\mathrm{MTL}}=-%
%TCIMACRO{\dsum \limits_{j}}%
%BeginExpansion
{\displaystyle\sum\limits_{j}}
%EndExpansion
\partial_{z}I_{b}V_{j}=-\left(  \partial_{z}I_{b}B,V\right)  =\left(
C\partial_{t}V,V\right)  +\left(  \partial_{z}I,V\right)  =\label{LLQq4}\\
=\partial_{t}\left[  \frac{1}{2}\left(  CV,V\right)  \right]  +\partial
_{z}\left(  I,V\right)  -\left(  I,\partial_{z}V\right)  =\partial_{t}\left[
\frac{1}{2}\left(  CV,V\right)  \right]  +\left(  L\partial_{t}I,I\right)
+\partial_{z}\left(  I,V\right)  =\nonumber\\
=\partial_{t}\left[  \frac{1}{2}\left(  CV,V\right)  +\frac{1}{2}\left(
LI,I\right)  \right]  +\partial_{z}(I,V).\nonumber
\end{gather}

\end{proof}

The first two terms in (\ref{exprpower1}) correspond to $\partial_{t}H$ where
\begin{equation}
H=\frac{1}{2}(CV,V)+\frac{1}{2}(LI,I) \label{LLQq5}%
\end{equation}
is the density of the total energy stored in the shunt capacitors and the
inductances per unit length. The last term in $P_{\mathrm{B}\rightarrow
\mathrm{MTL}}$ represents the divergence of the energy flux, $S=(I,V).$ In the
particular case of one line, we recover the usual expressions for the
corresponding quantities:%
\begin{equation}
P_{\mathrm{B}\rightarrow\mathrm{MTL}}=\partial_{t}\left[  \frac{1}{2}%
CV^{2}\right]  +\partial_{t}\left[  \frac{1}{2}LI^{2}\right]  +\partial
_{z}(IV). \label{PowerOneLine}%
\end{equation}
Our next result deals with the direction on the (time averaged) power flow.

\begin{theorem}
Let $k_{0},v_{0}$ denote the complex values of the wave number and the
velocity for the unique \emph{\ }exponentially growing solution according to
Theorem \ref{TeoremAmplification}. Then, the following formula holds for the
time average of the power:%
\begin{equation}
\left\langle P_{\mathrm{B}\rightarrow\mathrm{MTL}}\right\rangle (z)=-\left[
\omega\xi\left\vert k_{0}\right\vert ^{2}\left\vert \widehat{q}\right\vert
^{2}(\operatorname*{Re}v_{0}-u_{0})\operatorname{Im}v_{0}\right]
\mathrm{e}^{-2\left(  \operatorname{Im}k_{0}\right)  z}. \label{exprpower2}%
\end{equation}
Moreover, $\left\langle P_{\mathrm{B}\rightarrow\mathrm{MTL}}\right\rangle
(z)>0$ for all $z.$Thus, the power on the growing solution flows from the beam
to the MTL.
\end{theorem}

\begin{proof}
First, observe that for real time harmonic solutions $Q$ and $q$ of the form%
\begin{equation}
Q=\operatorname*{Re}\left(  \widehat{Q}\mathrm{e}^{i(kz-\omega t)}\right)
,\qquad q=\operatorname*{Re}\left(  \widehat{q}\mathrm{e}^{i(kz-\omega
t)}\right)  , \label{gblag16}%
\end{equation}
where $\widehat{Q}_{,}\widehat{q}$ are complex constants, the expression for
$P_{\mathrm{B}\rightarrow\mathrm{MTL}}$ can be written in the form%
\begin{equation}
P_{\mathrm{B}\rightarrow\mathrm{MTL}}=\partial_{;z}Q^{T}C^{-1}B\partial
_{tz}^{2}q=\operatorname*{Re}(\widehat{a}(z)e^{-i\omega t})\operatorname*{Re}%
(\widehat{b}(z)e^{-i\omega t}), \label{gblag17}%
\end{equation}
where%
\begin{equation}
\widehat{a}(z)=ike^{ikz}(\widehat{Q}+B\widehat{q})^{T}C^{-1};\qquad
\widehat{b}(z)=\omega k\widehat{q}e^{ikz}B. \label{gbalg18}%
\end{equation}
Applying formula (\ref{encoq5b}) for time average, we get
\begin{equation}
\left\langle P_{\mathrm{B}\rightarrow\mathrm{MTL}}\right\rangle (z)=\frac
{\omega}{2}\mathrm{e}^{-2\left(  \operatorname{Im}k\right)  z}%
\operatorname{Im}\left\{  \left\vert k\right\vert ^{2}\left(  \widehat{Q}%
+B\widehat{q}\right)  ^{\ast T}C^{-1}B\widehat{q}\right\}  . \label{Pow1}%
\end{equation}
Suppose now that $k_{0}$ is the complex root providing amplification, that is,
in the notation of Subsection \ref{BehAmplificationFactor}, $k_{0}%
=\omega/v_{0}$ with $\operatorname{Im}k_{0}<0.$ Then, $v_{0}$ is a root of the
system (\ref{Systemforv}) and therefore, in the notation of Section
\ref{AmplMTL-beam} and returning to the variable $k$,%
\[
k_{0}^{2}(\widehat{Q}^{\mathrm{T}}D+\widehat{q}d)=\xi(\omega-k_{0}u_{0}%
)^{2}\widehat{q}.
\]
Taking complex conjugate in the above equation and observing that $C^{-1}B=D$
and $B^{\mathrm{T}}C^{-1}B=d,$we can rewrite (\ref{Pow1}) in the form%
\begin{align}
\left\langle P_{\mathrm{B}\rightarrow\mathrm{MTL}}\right\rangle (z)  &
=\frac{\omega\xi}{2}\mathrm{e}^{-2\left(  \operatorname{Im}k_{0}\right)
z}\operatorname{Im}\left\{  \frac{\left\vert k_{0}\right\vert ^{2}}%
{k_{0}^{\ast2}}\left(  \omega-u_{0}k_{0}^{\ast}\right)  ^{2}\left\vert
\widehat{q}\right\vert ^{2}\right\} \label{Pow2}\\
&  =\frac{\omega\xi\left\vert k_{0}\right\vert ^{2}\left\vert \widehat{q}%
\right\vert ^{2}u_{0}^{2}}{2}\mathrm{e}^{-2\left(  \operatorname{Im}%
k_{0}\right)  z}\operatorname{Im}\left\{  \left(  \frac{k_{b}-k_{0}^{\ast}%
}{k_{0}^{\ast}}\right)  ^{2}\right\}  ,\qquad k_{b}=\frac{\omega}{u_{0}%
}.\nonumber
\end{align}
In terms of velocities, we have%
\begin{equation}
\operatorname{Im}\left(  \frac{k_{b}-k_{0}^{\ast}}{k_{0}^{\ast}}\right)
^{2}=\operatorname{Im}\left(  \frac{v_{0}^{\ast}}{u_{0}}-1\right)  ^{2}%
=-\frac{2}{u_{0}^{2}}(\operatorname*{Re}v_{0}-u_{0})\operatorname{Im}v_{0}.
\label{Pow3}%
\end{equation}
Since we are assuming $\operatorname{Im}v_{0}>0,$ we see from formula
(\ref{Pow2}) that $\left\langle P_{\mathrm{B}\rightarrow\mathrm{MTL}%
}\right\rangle (z)\geq0$ for all $z$ exactly if $\operatorname*{Re}v_{0}\leq
u_{0}$. But this is always the case, as it follows from (\ref{Vieta1}) and
(\ref{Asymmetry}). Formula (\ref{exprpower2}) follows at once from
(\ref{Pow2}) and (\ref{Pow3}).
\end{proof}

Observe also that, since $\operatorname{Im}k_{0}<0$, formula (\ref{exprpower2}%
) implies that $\left\langle P_{\mathrm{B}\rightarrow\mathrm{MTL}%
}\right\rangle $ increases in the $+z$ direction. For the evanescent wave,
corresponding to the value $k_{0}^{\ast},$ we have exactly the opposite
situation: the energy flows from the MTL to the beam and the power flux
decreases in the $+z$ direction.

\section{The Pierce model revisited\label{PierceRev copy(1)}}

Let us examine Pierce's original results in the light of our general theory.
They correspond to $n=1$, hence $d=D=C^{-1}$ and the dispersion relation
$\left\vert \widetilde{A}(v)\right\vert =0$ becomes%
\begin{equation}
\left(  -v^{2}L+C^{-1}\right)  \left[  C^{-1}-\xi(v-u_{0})^{2}\right]
-C^{-2}=0, \label{PierRev1}%
\end{equation}
which, in terms of $\ k=\omega/v,$ reads%
\begin{equation}
-L\omega^{2}k^{2}+\xi(\omega-ku_{0})^{2}(LC\omega^{2}-k^{2})=0.
\label{PierRev2}%
\end{equation}
After elementary algebraic transformations the above equation turns into%
\begin{equation}
u_{0}^{2}k^{4}-2u_{0}\omega k^{3}+\left[  1+\frac{L}{\xi}-LCu_{0}^{2}\right]
\omega^{2}k^{2}+2LCu_{0}\omega^{3}k-LC\omega^{4}=0, \label{PierRev3}%
\end{equation}
which is precisely the fourth order equation in \cite[ (1.16)]{Pier51}

The TL has only two characteristic velocities, namely $\pm v_{1}=\pm
1/\sqrt{LC}$ which are not solutions of (\ref{PierRev1}). The graph of the
characteristic function $R$ has only two vertical asymptotes at $v=\pm v_{1}$.
The special regime considered in \cite{Pier51} corresponds to taking large
$\xi$, and $u_{0}=v_{1}$. As we know, in this case amplification occurs for
any $\xi>0$. For small values of the parameter%
\begin{equation}
k_{p}=\frac{\omega_{p}}{u_{0}}=\frac{1}{u_{0}}\sqrt{\frac{4\pi}{\sigma\xi}},
\label{kpuu1}%
\end{equation}
Pierce asserts that $k\simeq k_{\mathrm{b}}=\omega/u_{0}$ for the forward
unattenuated wave. In terms of velocities this means that for large values of
$\xi$ the positive real solution $v_{1}^{+}$ is very close to $u_{0}$. The
graph in Figure \ref{PierceAmp} refers to this situation and it clearly shows
that indeed $v_{1}^{+}\simeq u_{0}$ and $-v_{1}^{-}$ $\simeq-u_{0}$ for large
$\ \xi$ (the parabola becomes very narrow and the right and left branches of
the graph of $R$ are intersected close to the asymptotes).%
%TCIMACRO{\FRAME{ftbpFU}{3.2543in}{2.4102in}{0pt}{\Qcb{Pierce's dispersion
%relation for $u_{0}=v_{1}$: $\ $For large $\xi,$ the parabola $y=-\xi
%(v-u_{0})^{2}$ is very narrow and intersects the graph of $y=R(v)$ close to
%the asymptotes: $v_{1}^{+},v_{1}^{-}\approx u_{0}.$}}{\Qlb{PierceAmp}%
%}{pierceamp.eps}{\special{ language "Scientific Word";  type "GRAPHIC";
%display "PICT";  valid_file "F";  width 3.2543in;  height 2.4102in;
%depth 0pt;  original-width 8.2229in;  original-height 8.2229in;
%cropleft "0";  croptop "1";  cropright "1";  cropbottom "0";
%filename 'PierceAmp.eps';file-properties "XNPEU";}} }%
%BeginExpansion
\begin{figure}[ptb]%
\centering
\ifcase\msipdfoutput
\includegraphics[
height=2.4102in,
width=3.2543in
]%
{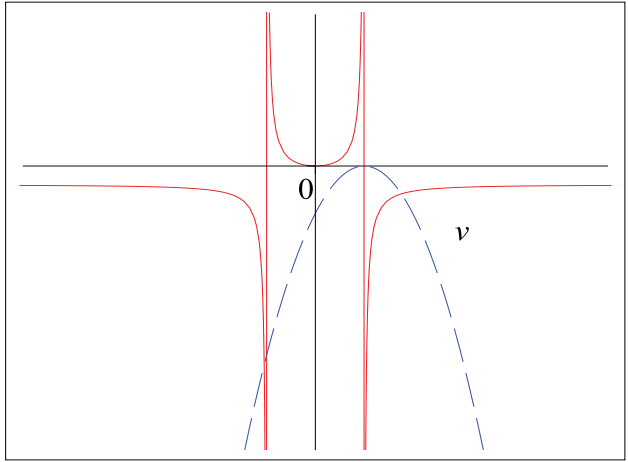}%
\else
\includegraphics[
height=2.4102in,
width=3.2543in
]%
{D:/alatex/Preparation/Reyes/lagsys/arxiv2/graphics/PierceAmp__3.pdf}%
\fi
\caption{Pierce's dispersion relation for $u_{0}=v_{1}$: $\ $For large $\xi,$
the parabola $y=-\xi(v-u_{0})^{2}$ is very narrow and intersects the graph of
$y=R(v)$ close to the asymptotes: $v_{1}^{+},v_{1}^{-}\approx u_{0}.$}%
\label{PierceAmp}%
\end{figure}
%EndExpansion

Consequently, the identity%
\begin{equation}
2\operatorname*{Re}v_{0}+v_{1}^{+}+v_{1}^{-}=2u_{0} \label{kpuu2}%
\end{equation}
implies that $\operatorname*{Re}v_{0}$, $\operatorname*{Re}v_{0}^{\ast}\simeq
u_{0}$. Therefore, three solutions have real part close to $u_{0}$ and the
remaining real solution is close to $-u_{0}.$The latter corresponds to the
backward wave. In terms of the wavenumber, three solutions have real part
close to $k_{\mathrm{b}}$. By looking for solutions (in $k)$ to the fourth
order equation (\ref{PierRev3}) in the form%
\[
k=k_{\mathrm{b}}+\mathrm{i}\delta,
\]
with small (compared to $k_{\mathrm{b}}$) complex $\delta$, Pierce gets rid of
the backward wave. The dispersion relation (\ref{PierRev2}) in terms of
$\delta$ reads%
\begin{equation}
\left(  \mathrm{i}\delta\right)  ^{3}\left(  2+\mathrm{i}\delta k_{\mathrm{b}%
}^{-1}\right)  =-L\xi^{-1}k_{\mathrm{b}}^{2}\left(  1+\mathrm{i}\delta
k_{\mathrm{b}}^{-1}\right)  ^{2}. \label{kp7uu3}%
\end{equation}
Neglecting $\mathrm{i}\delta/k_{\mathrm{b}}$ we arrive at Pierce's third
degree equation for $\delta$:%
\begin{equation}
\delta^{3}=-\frac{Lk_{\mathrm{b}}^{2}\xi^{-1}}{2}\mathrm{i}, \label{kpuu4}%
\end{equation}
which has three complex roots,%
\begin{equation}
\delta_{1}=c\mathrm{i},\qquad\delta_{2}=c\left(  -\sqrt{3}-\mathrm{i}\right)
/2,\qquad\delta_{3}=c\left(  \sqrt{3}-\mathrm{i}\right)  /2,\text{ where
}c=\sqrt[3]{Lk_{\mathrm{b}}^{2}\xi^{-1}/2}, \label{kpuu5}%
\end{equation}
corresponding respectively to the unattenuated wave faster than the natural
phase velocity of the circuit ($v_{1}^{+}>v_{1}=u_{0})$, the increasing and
the decreasing waves.

It is clear from the analysis in Section \ref{AppAmplification} that in case
of several identical, non-interacting TLs, only two asymptotes are present in
the graph of $R$. This fact suggests that we can replace such a system by a
single effective line, with modified parameters, interacting with the beam.
Indeed, let $C=\widehat{C}\mathrm{Id}_{n}$, $L=\widehat{L}\mathrm{Id}_{n}$
with $n\geq2$. Then, there are exactly two characteristic velocities $\pm
v_{1}$ where $v_{1}=1/\sqrt{\widehat{L}\widehat{C}}$. According to Section
\ref{AppAmplification} the latter are necessarily characteristic velocities of
the entire system, of multiplicity $n-1$ each. Using the notation from that
section, we have%
\[
D=\widehat{C}^{-1}(1,1,...1)^{T},\text{ \ \ \ }d=n\widehat{C}^{-1},\text{
\ \ \ \ }\widetilde{D}=\widehat{L}^{-1/2}\widehat{C}^{-1}(1,1,...1)^{T}%
.\text{\ }%
\]
The MTL characteristic function $R(v)$ has the explicit expression%
\begin{equation}
R(v)=\frac{n\widehat{L}^{-1}\widehat{C}^{-2}}{v_{1}^{2}-v^{2}}-n\widehat{C}%
^{-1}. \label{kpuu6}%
\end{equation}
If we choose $\widetilde{C}=\widehat{C}/n$ and $\widetilde{L}=n\widehat{L},$
the above function coincides with the characteristic function for one line
with parameters $\widetilde{C}$ and $\widetilde{L},$%
\[
R(v)=\frac{\widetilde{L}^{-1}\widetilde{C}^{-2}}{v_{1}^{2}-v^{2}%
}-\widetilde{C}^{-1}.
\]
Since amplification depends only on the complex root of the dispersion
relation, which is a root of the canonical dispersion relation, amplification
factors also coincide.

Actually, a more general assertion holds:

\begin{theorem}
Let $C$ and $L$ be the capacity, respectively inductance matrices of an
$n$-lines MTL. If
\begin{equation}
LC=v_{1}^{-2}Id, \label{linesequivalence}%
\end{equation}
then the canonical dispersion relation of the system consisting of \ the MTL
and a given beam coincides with the canonical dispersion relation of the
system consisting of a single transmission line with parameters $\widetilde{L}%
,\widetilde{C}$ defined by
\[
\widetilde{C}^{-1}=\sum_{i,j=1}^{n}(C^{-1})_{ij},\qquad\widetilde{L}%
=v_{1}^{-2}\widetilde{C}^{-1}%
\]
and the same beam. Consequently, the amplification factors of both systems coincide.
\end{theorem}

It should be noted that the multiple line system and the reduced (one line)
system above are not equivalent in all respects. Actually, the multiline
system admits oscillatory modes with eigenvelocity $\pm v_{1},$ whereas the
equivalent one-line system does not. However, the exponentially growing and
evanescent modes coincide, as well as the two purely oscillatory modes with
eigenvelocities different from $\pm v_{1}.$ The proof of the above theorem is
a straightforward generalization of the case of identical lines and we omit it.

A different reduction can be achieved by suitably modifying the beam. Suppose
we have $n$ identical, uncoupled lines as before. Dividing \ the dispersion
relation (\ref{kpuu6}) by $n,$ we conclude that the interaction of the system
with a beam with parameters $\left(  \xi,u_{0}\right)  $ is equivalent to the
interaction of one line with parameters $\widehat{L}$, $\widehat{C}$ with a
beam with parameters $\left(  \xi/n,u_{0}\right)  $. The asymptotic formula
(\ref{AmpFactorAt0}) then implies that the amplification factor grows like
$\sqrt{n}$ as $n\rightarrow\infty$.

\section{Mathematical subjects\label{MathSubj}}

\subsection{ de Donder-Weyl version of the Hamiltonian
formalism\label{AppdeDonder}}

In this section we introduce basic settings of the de Donder-Weyl (DW) version
of the Hamilton equations which treats the time and space variable in equal
manner just as the Lagrangian approach which constitutes its basis. The DW
theory is a generalization of the standard Hamiltonian formalism and the
Hamilton-Jacobi theory, \cite[4.2]{Rund} that has the advantage of requiring a
finite-dimensional phase space. We do not use any significant results of the
DW theory but rather take advantage of its set up that allows to treat the
time $t$ and the space variable $z$ on equal footing. We remind that the
standard Hamilton-Jacobi theory gives preferential treatment to time $t$.

Let us consider a system $\mathsf{q}=\left\{  \mathsf{q}_{j}\left(
t,z\right)  ,\ j=1,\ldots n\right\}  $ of real valued fields depending on time
$t$ and one-dimensional space variable $z$. Suppose it has a Lagrangian
density of the form%
\begin{equation}
\mathcal{L}=\mathcal{L}\left(  t,z,\mathsf{q},\mathsf{q}_{,t},\mathsf{q}%
_{,z}\right)  ,\text{ where }\mathsf{q}_{,t}=\partial_{t}\mathsf{q}%
,\ \mathsf{q}_{,z}=\partial_{z}\mathsf{q}. \label{donwe1}%
\end{equation}
The corresponding Euler-Lagrange equations are, \cite[4.16]{GelFom}%
\begin{equation}
\frac{\partial\mathcal{L}}{\partial\mathsf{q}}-\partial_{t}\frac
{\partial\mathcal{L}}{\partial\mathsf{q}_{,t}}-\partial_{z}\frac
{\partial\mathcal{L}}{\partial\mathsf{q}_{,z}}=0. \label{donwe2}%
\end{equation}
Evidently, (\ref{donwe2}) is a system of second order partial differential
equations for $\mathsf{q}$ as a function of $t,z$. It can be recast as a first
order partial differential system with respect to time $t$ or with respect to
the space variable $z$ using a generalization of the standard Hamiltonian
formalism known as de Donder-Weyl (DW) theory. Thus, following the DW theory
we introduce two canonical momenta densities $\mathsf{p}_{t}$ and
$\mathsf{p}_{z}$ and the DW Hamiltonian density $\mathcal{H}$ by the formulas%
\begin{gather}
\mathsf{p}_{t}=\frac{\partial\mathcal{L}}{\partial\mathsf{q}_{,t}}\left(
t,z,\mathsf{q},\mathsf{q}_{,t},\mathsf{q}_{,z}\right)  ,\label{donwe3a}\\
\mathsf{p}_{z}=\frac{\partial\mathcal{L}}{\partial\mathsf{q}_{,z}}\left(
t,z,\mathsf{q},\mathsf{q}_{,t},\mathsf{q}_{,z}\right)  ,\label{donwe3b}\\
\mathcal{H}_{\mathrm{DW}}=\mathcal{H}_{\mathrm{DW}}\left(  t,z,\mathsf{q}%
,\mathsf{p}_{t},\mathsf{p}_{z}\right)  =\mathsf{p}_{t}^{\mathrm{T}}%
\mathsf{q}_{,t}+\mathsf{p}_{z}^{\mathrm{T}}\mathsf{q}_{,z}-\mathcal{L}\left(
t,z,\mathsf{q},\mathsf{q}_{,t},\mathsf{q}_{,z}\right)  , \label{donwe3c}%
\end{gather}
where $\mathsf{q}_{,t}$ and $\mathsf{q}_{,z}$ are supposed to be found from
respective equations (\ref{donwe3a})-(\ref{donwe3b}) and to be substituted in
the right-hand side for the second equation in (\ref{donwe3c}). Then the
corresponding DW version of the Hamilton equations are%
\begin{gather}
\partial_{t}\mathsf{q}=\frac{\partial\mathcal{H}_{\mathrm{DW}}}{\partial
\mathsf{p}_{t}}\left(  t,z,\mathsf{q},\mathsf{p}_{t},\mathsf{p}_{z}\right)
,\label{donwe4a}\\
\partial_{z}\mathsf{q}=\frac{\partial\mathcal{H}_{\mathrm{DW}}}{\partial
\mathsf{p}_{z}}\left(  t,z,\mathsf{q},\mathsf{p}_{t},\mathsf{p}_{z}\right)
,\label{donwe4b}\\
\partial_{t}\mathsf{p}_{t}+\partial_{z}\mathsf{p}_{z}=-\frac{\partial
\mathcal{H}_{\mathrm{DW}}}{\partial\mathsf{q}}\left(  t,z,\mathsf{q}%
,\mathsf{p}_{t},\mathsf{p}_{z}\right)  , \label{donwe4c}%
\end{gather}
and this system of $3n$ first order equations is equivalent to the
Euler-Lagrange system (\ref{donwe2}).

One can solve the system (\ref{donwe3a}) -(\ref{donwe3b}) for $\mathsf{q}%
_{,t}$ and $\mathsf{q}_{,z}$ in terms of the momenta, obtaining
representations
\begin{equation}
\mathsf{q}_{,t}=G_{t}\left(  t,z,\mathsf{q},\mathsf{p}_{t},\mathsf{p}%
_{z}\right)  ,\qquad\mathsf{q}_{,z}=G_{z}\left(  t,z,\mathsf{q},\mathsf{p}%
_{t},\mathsf{p}_{z}\right)  , \label{donwe5a}%
\end{equation}
for some functions $G_{t}$ and $G_{z}$. Solving for $\mathsf{p}_{t}$ in the
first and for $\mathsf{p}_{z}$ in the second, we get
\begin{equation}
\mathsf{p}_{t}=K_{t}\left(  t,z,\mathsf{q},\mathsf{q}_{,t},\mathsf{p}%
_{z}\right)  ,\qquad\mathsf{p}_{z}=K_{z}\left(  t,z,\mathsf{q},\mathsf{p}%
_{t},\mathsf{q}_{,z}\right)  , \label{donwe5b}%
\end{equation}
for some functions $K_{t}$ and $K_{z}$.

To obtain the first order partial differential equations with respect to $t$
we consider the pair $\mathsf{p}_{t},\mathsf{q}$ and using equations
(\ref{donwe4a}) and (\ref{donwe4c}) we get%
\begin{gather}
\partial_{t}\mathsf{q}=\frac{\partial H_{\mathrm{DW}}}{\partial\mathsf{p}_{t}%
}\left(  t,z,\mathsf{q},\mathsf{p}_{t},\mathsf{p}_{z}\right)  =F_{\mathsf{q}%
}\left(  t,z,\mathsf{q},\mathsf{q}_{,z},\mathsf{p}_{t}\right)
,\label{donwe6a}\\
\partial_{t}\mathsf{p}_{t}=-\partial_{z}\mathsf{p}_{z}-\frac{\partial
H_{\mathrm{DW}}}{\partial\mathsf{q}}\left(  t,z,\mathsf{q},\mathsf{p}%
_{t},\mathsf{p}_{z}\right)  =F_{\mathsf{p}}\left(  t,z,\mathsf{q}%
,\mathsf{q}_{,z},\mathsf{q}_{,zz},\mathsf{p}_{t},\mathsf{p}_{t,z}\right)
,\nonumber
\end{gather}
where the expressions $F_{\mathsf{q}}$ and $F_{\mathsf{p}}$ are obtained by
replacing $\mathsf{p}_{z}$ in (\ref{donwe6a}) by its representation
(\ref{donwe5b}). Observe that the system of partial differential equations
(\ref{donwe6a}) for $\mathsf{p}_{t}$ and $\mathsf{q}$ is of the first order
with respect to time $t$.

To obtain the first order partial differential equations with respect to $z$
we consider the pair $\mathsf{p}_{z},\mathsf{q}$ and\ proceed just as in the
previous case with using the equations (\ref{donwe4b}) and (\ref{donwe4c}) to
get%
\begin{gather}
\partial_{z}\mathsf{q}=\frac{\partial H_{\mathrm{DW}}}{\partial\mathsf{p}_{z}%
}\left(  t,z,\mathsf{q},\mathsf{p}_{t},\mathsf{p}_{z}\right)  =\widetilde{F}%
_{\mathsf{q}}\left(  t,z,\mathsf{q},\mathsf{q}_{,t},\mathsf{p}_{z}\right)
,\label{donwe6b}\\
\partial_{z}\mathsf{p}_{z}=-\partial_{t}\mathsf{p}_{t}-\frac{\partial
H_{\mathrm{DW}}}{\partial\mathsf{q}}\left(  t,z,\mathsf{q},\mathsf{p}%
_{t},\mathsf{p}_{z}\right)  =\widetilde{F}_{\mathsf{p}}\left(  t,z,\mathsf{q}%
,\mathsf{q}_{,t},\mathsf{q}_{,tt},\mathsf{p}_{z},\mathsf{p}_{z,t}\right)
,\nonumber
\end{gather}
where the expressions $\widetilde{F}_{\mathsf{q}}$ and $\widetilde{F}%
_{\mathsf{p}}$ are determined by replacing $\mathsf{p}_{t}$ in the relevant
expressions in (\ref{donwe6b}) by its representation (\ref{donwe5b}). Observe
that the system of partial differential equations (\ref{donwe6b}) for
$\mathsf{p}_{z}$ and $\mathsf{q}$ is of the first order with respect to the
space variable $z$.

Summing up, we have proved the following

\begin{theorem}
The second order Euler-Lagrange system (\ref{donwe2}) is equivalent to either
the first order system (\ref{donwe6a}) for $\mathsf{q}$ and $\mathsf{p}_{t}$
or the first order system (\ref{donwe6b}) for $\mathsf{q}$ and $\mathsf{p}%
_{z}$.
\end{theorem}

\subsection{ Quadratic Lagrangian densities\label{AppQuadLag}}

In this section we present some results concerning a special family of
Lagrangians, namely those quadratic in the derivatives (and independent both
of coordinates and the fields). \ This kind of Lagrangians often appear in
practice, in particular in the TL-beam interaction system. Thus, let us
consider a quadratic Lagrangian density of the form%
\begin{equation}
\mathcal{L(}\mathsf{q}_{,t},\mathsf{q}_{,z})=\frac{1}{2}\partial_{t}%
\mathsf{q}^{\mathrm{T}}\alpha\partial_{t}\mathsf{q}+\partial_{t}%
\mathsf{q}^{\mathrm{T}}\theta\partial_{z}\mathsf{q}-\frac{1}{2}\partial
_{z}\mathsf{q}^{\mathrm{T}}\eta\partial_{z}\mathsf{q}, \label{qulaq1}%
\end{equation}
where $\mathsf{q}=\left\{  \mathsf{q}_{j}\left(  t,z\right)  ,\ j=1,\ldots
n\right\}  $ are real valued fields depending on time $t$ and one-dimensional
space variable $z$, $\mathsf{q}_{,t}=\partial_{t}\mathsf{q,}$ $\mathsf{q}%
_{,z}=\partial_{z}\mathsf{q}$ and $\alpha\left(  t,z\right)  $, $\eta\left(
t,z\right)  $, $\theta\left(  t,z\right)  $ are symmetric $n\times n$ matrices
with real entries, that is
\begin{equation}
\alpha^{\mathrm{T}}=\alpha,\qquad\eta^{\mathrm{T}}=\eta,\qquad\theta
^{\mathrm{T}}=\theta. \label{qulaq2}%
\end{equation}
The Lagrangian density (\ref{qulaq1}) can be recast into the following form,
involving a block matrix:%
\begin{equation}
\mathcal{L}=\frac{1}{2}\mathsf{u}^{\mathrm{T}}M_{\mathrm{L}}\mathsf{u;}\qquad
M_{\mathrm{L}}=\left[
\begin{array}
[c]{ll}%
\alpha & \theta\\
\theta & -\eta
\end{array}
\right]  ,\qquad\mathsf{u}=\left[
\begin{array}
[c]{l}%
\partial_{t}\mathsf{q}\\
\partial_{z}\mathsf{q}%
\end{array}
\right]  . \label{qulaq3}%
\end{equation}
The Euler-Lagrange equation (\ref{donwe2}) for this Lagrangian is
\begin{equation}
\left[  \partial_{t}\alpha\partial_{t}+\partial_{t}\theta\partial_{z}%
+\partial_{z}\theta\partial_{t}-\partial_{z}\eta\partial_{z}\right]
\mathsf{q}=0. \label{qulaq4}%
\end{equation}

Now we would like to use the DW Hamiltonian approach from the previous section
to recast the second order differential $n\times n$ system (\ref{qulaq4}) into
first order ones with respect to $t$ and with respect to $z$ as well. With
that in mind we introduce the canonical momenta as in (\ref{donwe3a}%
)-(\ref{donwe3b})%
\begin{equation}
\mathsf{p}_{t}=\frac{\partial\mathcal{L}}{\partial\mathsf{q}_{,t}}%
=\alpha\partial_{t}\mathsf{q}+\theta\partial_{z}\mathsf{q},\qquad
\mathsf{p}_{z}=\frac{\partial\mathcal{L}}{\partial\mathsf{q}_{,z}}%
=\theta\partial_{t}\mathsf{q}-\eta\partial_{z}\mathsf{q}, \label{qulaq5}%
\end{equation}
which can be recast as%
\begin{equation}
\mathsf{p}=\left[
\begin{array}
[c]{l}%
\mathsf{p}_{t}\\
\mathsf{p}_{z}%
\end{array}
\right]  =\left[
\begin{array}
[c]{ll}%
\alpha & \theta\\
\theta & -\eta
\end{array}
\right]  \left[
\begin{array}
[c]{l}%
\partial_{t}\mathsf{q}\\
\partial_{z}\mathsf{q}%
\end{array}
\right]  =M_{\mathrm{L}}\mathsf{u}, \label{qulaq5a}%
\end{equation}
or%
\begin{equation}
\left[
\begin{array}
[c]{l}%
\partial_{t}\mathsf{q}\\
\partial_{z}\mathsf{q}%
\end{array}
\right]  =\mathsf{u}=M_{\mathrm{L}}^{-1}\mathsf{p}=M_{\mathrm{L}}^{-1}\left[
\begin{array}
[c]{l}%
\mathsf{p}_{t}\\
\mathsf{p}_{z}%
\end{array}
\right]  . \label{qulaq5b}%
\end{equation}
Notice that the difference in signs in expressions for momenta $\mathsf{p}%
_{t}$ and $\mathsf{p}_{z}$ in (\ref{qulaq5}) is due to difference in signs for
matrices $\alpha$ and $\eta$ as they enter the expressions for the kinetic and
potential energies in the Lagrangian density defined by (\ref{qulaq1}).

Solving equations (\ref{qulaq5a}) for $\partial_{t}\mathsf{q}$ and
$\partial_{z}\mathsf{q}$ we obtain%
\begin{equation}
\partial_{t}\mathsf{q}=\alpha^{-1}\left(  \mathsf{p}_{t}-\theta\partial
_{z}\mathsf{q}\right)  ,\qquad\partial_{z}\mathsf{q}=\eta^{-1}\left(
\theta\partial_{t}\mathsf{q}-\mathsf{p}_{z}\right)  . \label{qulaq5c}%
\end{equation}
Using (\ref{qulaq1}) and (\ref{qulaq5}) we get the following identity
\begin{gather}
\mathsf{p}_{t}^{\mathrm{T}}\partial_{t}\mathsf{q}+\mathsf{p}_{z}^{\mathrm{T}%
}\partial_{z}\mathsf{q}=\partial_{t}\mathsf{q}^{\mathrm{T}}\mathsf{p}%
_{t}+\partial_{z}\mathsf{q}^{\mathrm{T}}\mathsf{p}_{z}=\label{qulaq5ca}\\
=\partial_{t}\mathsf{q}^{\mathrm{T}}\left(  \alpha\partial_{t}\mathsf{q}%
+\theta\partial_{z}\mathsf{q}\right)  +\partial_{z}\mathsf{q}^{\mathrm{T}%
}\left(  \theta\partial_{t}\mathsf{q}-\eta\partial_{z}\mathsf{q}\right)
=2L.\nonumber
\end{gather}
Then in view of (\ref{qulaq5ca}) the general DW Hamiltonian $\mathcal{H}%
_{\mathrm{DW}}$ defined by (\ref{donwe3c}) takes here the form%
\begin{equation}
\mathcal{H}_{\mathrm{DW}}=\mathsf{p}_{t}^{\mathrm{T}}\partial_{t}%
\mathsf{q}+\mathsf{p}_{z}^{\mathrm{T}}\partial_{z}\mathsf{q}-\mathcal{L}%
=\mathcal{L}=\frac{1}{2}\partial_{t}\mathsf{q}^{\mathrm{T}}\alpha\partial
_{t}\mathsf{q}+\partial_{t}\mathsf{q}^{\mathrm{T}}\theta\partial_{z}%
\mathsf{q}-\frac{1}{2}\partial_{z}\mathsf{q}^{\mathrm{T}}\eta\partial
_{z}\mathsf{q}. \label{qulaq5cb}%
\end{equation}
Another way to obtain a representation for the DW Hamiltonian is to use
(\ref{qulaq5b}) yielding%
\begin{equation}
\mathcal{H}_{\mathrm{DW}}=\mathsf{p}^{\mathrm{T}}\mathsf{u}-\frac{1}%
{2}\mathsf{u}^{\mathrm{T}}M_{\mathrm{L}}\mathsf{u}=\mathsf{u}^{\mathrm{T}%
}M_{\mathrm{L}}\mathsf{u}-\frac{1}{2}\mathsf{u}^{\mathrm{T}}M_{\mathrm{L}%
}\mathsf{u}=\frac{1}{2}\mathsf{u}^{\mathrm{T}}M_{\mathrm{L}}\mathsf{u}%
=\mathcal{L}=\frac{1}{2}\mathsf{p}^{\mathrm{T}}M_{\mathrm{L}}^{-1}\mathsf{p}.
\label{qulaq5d}%
\end{equation}
Observe that the DW Hamiltonian $\mathcal{H}$ equals the Lagrangian
$\mathcal{L}$ at the corresponding point, that is%
\begin{equation}
\mathcal{H}_{\mathrm{DW}}=\mathcal{H}_{\mathrm{DW}}\left(  \mathsf{p}\right)
=\mathcal{L}\left(  \mathsf{u}\right)  =\mathcal{L},\text{ where }%
\mathsf{p}=M_{\mathrm{L}}\mathsf{u} \label{qulaq5da}%
\end{equation}
(actually, this is a general property of the Legendre transform of homogeneous
quadratic polynomials). The equation (\ref{donwe4c}) takes here the form%
\begin{equation}
\partial_{t}\mathsf{p}_{t}+\partial_{z}\mathsf{p}_{z}=0. \label{qulaq5e}%
\end{equation}

To obtain the first order equations with respect to $t$ we pick the pair
$\mathsf{p}_{t},\mathsf{q}$. We use equations (\ref{qulaq5e}) and
(\ref{qulaq5c}) for respectively $\partial_{t}\mathsf{p}_{t}$ and
$\partial_{t}\mathsf{q}$. We eliminate $\mathsf{p}_{z}$ in (\ref{qulaq5e}) by
using its representation (\ref{qulaq5}) getting the system%
\begin{gather}
\partial_{t}\mathsf{p}_{t}=-\partial_{z}\mathsf{p}_{z}=-\partial_{z}%
\theta\partial_{t}\mathsf{q}+\partial_{z}\eta\partial_{z}\mathsf{q}%
=-\partial_{z}\theta\alpha^{-1}\left(  \mathsf{p}_{t}-\theta\partial
_{z}\mathsf{q}\right)  +\partial_{z}\eta\partial_{z}\mathsf{q},
\label{qulaq6a}\\
\partial_{t}\mathsf{q}=\alpha^{-1}\left(  \mathsf{p}_{t}-\theta\partial
_{z}\mathsf{q}\right)  . \label{qulaq6b}%
\end{gather}
Observe that we used equation (\ref{qulaq6b}) to get the right-hand side of
equation (\ref{qulaq6a}). The above system can be written in matrix form
\begin{equation}
\partial_{t}\left[
\begin{array}
[c]{l}%
\mathsf{p}_{t}\\
\mathsf{q}%
\end{array}
\right]  =\left[
\begin{array}
[c]{ll}%
-\partial_{z}\theta\alpha^{-1} & \partial_{z}\eta\partial_{z}+\partial
_{z}\theta\alpha^{-1}\theta\partial_{z}\\
\alpha^{-1} & -\alpha^{-1}\theta\partial_{z}%
\end{array}
\right]  \left[
\begin{array}
[c]{l}%
\mathsf{p}_{t}\\
\mathsf{q}%
\end{array}
\right]  . \label{qulaq6c}%
\end{equation}
One can recast the above system into a canonical Hamiltonian form by using the
following symplectic matrix%
\begin{equation}
J=\left[
\begin{array}
[c]{cc}%
0 & -\mathbf{1}\\
\mathbf{1} & 0
\end{array}
\right]  ,\quad J^{2}=-\mathbf{1},\quad J=-J^{\mathrm{T}}. \label{qulaq6d}%
\end{equation}
Namely,%
\begin{equation}
\partial_{t}V=JM_{\mathrm{Ht}}V,\quad V=\left[
\begin{array}
[c]{l}%
\mathsf{p}_{t}\\
\mathsf{q}%
\end{array}
\right]  \label{qulaq6e}%
\end{equation}
where%
\begin{gather}
M_{\mathrm{Ht}}=\left[
\begin{array}
[c]{ll}%
\alpha^{-1} & -\alpha^{-1}\theta\partial_{z}\\
\partial_{z}\theta\alpha^{-1} & -\partial_{z}\theta\alpha^{-1}\theta
\partial_{z}-\partial_{z}\eta\partial_{z}%
\end{array}
\right]  =\label{qulaq6f}\\
=\left[
\begin{array}
[c]{cc}%
\mathbf{1} & 0\\
\partial_{z}\theta & \mathbf{1}%
\end{array}
\right]  \left[
\begin{array}
[c]{cc}%
\alpha^{-1} & 0\\
0 & -\partial_{z}\eta\partial_{z}%
\end{array}
\right]  \left[
\begin{array}
[c]{cc}%
\mathbf{1} & -\theta\partial_{z}\\
0 & \mathbf{1}%
\end{array}
\right]  .\nonumber
\end{gather}

To obtain the first order equations with respect to $z$ we pick the pair
$\mathsf{p}_{z},\mathsf{q}$. We use equations (\ref{qulaq5e}) and
(\ref{qulaq5c}) for respectively $\partial_{z}\mathsf{p}_{z}$ and
$\partial_{z}\mathsf{q}$. We eliminate $\mathsf{p}_{t}$ in (\ref{qulaq5e}) by
using its representation (\ref{qulaq5}) getting the system%
\begin{gather}
\partial_{z}\mathsf{p}_{z}=-\partial_{t}\mathsf{p}_{t}=-\partial_{t}\left(
\alpha\partial_{t}\mathsf{q}+\theta\partial_{z}\mathsf{q}\right)
=-\partial_{t}\alpha\partial_{t}\mathsf{q}-\partial_{t}\theta\eta^{-1}\left(
\theta\partial_{t}\mathsf{q}-\mathsf{p}_{z}\right)  ,\label{qulaq7}\\
\partial_{z}\mathsf{q}=\eta^{-1}\left(  \theta\partial_{t}\mathsf{q}%
-\mathsf{p}_{z}\right)  . \label{qulaq7a}%
\end{gather}
Observe that we used equation (\ref{qulaq7a}) to get the right-hand side of
equation (\ref{qulaq7}). The above system can be written as
\begin{equation}
\partial_{z}\left[
\begin{array}
[c]{l}%
\mathsf{p}_{z}\\
\mathsf{q}%
\end{array}
\right]  =\left[
\begin{array}
[c]{ll}%
\partial_{t}\theta\eta^{-1} & -\partial_{t}\alpha\partial_{t}-\partial
_{t}\theta\eta^{-1}\theta\partial_{t}\\
-\eta^{-1} & \eta^{-1}\theta\partial_{t}%
\end{array}
\right]  \left[
\begin{array}
[c]{l}%
\mathsf{p}_{z}\\
\mathsf{q}%
\end{array}
\right]  . \label{qulaq7b}%
\end{equation}
The system (\ref{qulaq7b}) can be transformed into the following canonical
Hamiltonian form%
\begin{equation}
\partial_{z}V=JM_{\mathrm{Hz}}V,\quad V=\left[
\begin{array}
[c]{l}%
\mathsf{p}_{z}\\
\mathsf{q}%
\end{array}
\right]  , \label{qulaq7c}%
\end{equation}
where%
\begin{gather}
M_{\mathrm{Hz}}=\left[
\begin{array}
[c]{ll}%
-\eta^{-1} & \eta^{-1}\theta\partial_{t}\\
-\partial_{t}\theta\eta^{-1} & \partial_{t}\alpha\partial_{t}+\partial
_{t}\theta\eta^{-1}\theta\partial_{t}%
\end{array}
\right]  =\label{qulaq7d}\\
=\left[
\begin{array}
[c]{cc}%
\mathbf{1} & 0\\
\partial_{t}\theta & \mathbf{1}%
\end{array}
\right]  \left[
\begin{array}
[c]{ll}%
-\eta^{-1} & 0\\
0 & \partial_{t}\alpha\partial_{t}%
\end{array}
\right]  \left[
\begin{array}
[c]{cc}%
\mathbf{1} & -\theta\partial_{t}\\
0 & \mathbf{1}%
\end{array}
\right]  .\nonumber
\end{gather}
Comparing expressions (\ref{qulaq6f}) and (\ref{qulaq7d}) we observe a
noticeable difference in signs that is explained by the difference in signs in
the expressions for the kinetic and potential energies in the Lagrangian
density defined by (\ref{qulaq1}).

We can transform the system (\ref{qulaq7c})-(\ref{qulaq7d}) further yet into
another form intimately related to the energy conservation law. For that we
begin with the identity%
\begin{gather}
M_{\mathrm{Hz}}=\left[
\begin{array}
[c]{ll}%
-\eta^{-1} & \eta^{-1}\theta\partial_{t}\\
-\partial_{t}\theta\eta^{-1} & \partial_{t}\alpha\partial_{t}+\partial
_{t}\theta\eta^{-1}\theta\partial_{t}%
\end{array}
\right]  =\label{mhzne1}\\
=\left[
\begin{array}
[c]{cc}%
\mathbf{1} & 0\\
0 & -\partial_{t}%
\end{array}
\right]  \left[
\begin{array}
[c]{ll}%
-\eta^{-1} & \eta^{-1}\theta\\
\theta\eta^{-1} & -\alpha-\theta\eta^{-1}\theta
\end{array}
\right]  \left[
\begin{array}
[c]{cc}%
\mathbf{1} & 0\\
0 & \partial_{t}%
\end{array}
\right]  .\nonumber
\end{gather}
Based on (\ref{mhzne1}), the system (\ref{qulaq7c})-(\ref{qulaq7d}) can be
recast into the following "Hamiltonian" form
\begin{equation}
\tilde{J}\partial_{z}V=\mathrm{i}\partial_{t}\tilde{M}V,\quad V=\left[
\begin{array}
[c]{l}%
\mathsf{p}_{z}\\
\partial_{t}\mathsf{q}%
\end{array}
\right]  , \label{mhzne2}%
\end{equation}
where%
\begin{equation}
\tilde{J}=\left[
\begin{array}
[c]{cc}%
0 & \mathrm{i}\mathbf{1}\\
\mathrm{i}\mathbf{1} & 0
\end{array}
\right]  ,\qquad\tilde{M}=\left[
\begin{array}
[c]{ll}%
-\eta^{-1} & \eta^{-1}\theta\\
\theta\eta^{-1} & -\alpha-\theta\eta^{-1}\theta
\end{array}
\right]  . \label{mhzne3}%
\end{equation}
When deriving the Hamiltonian equation (\ref{mhzne2})-(\ref{mhzne3}) we used
the following identity relating $\tilde{J}$ and $J$ defined in (\ref{qulaq6d})%
\begin{equation}
\left[
\begin{array}
[c]{cc}%
\mathbf{1} & 0\\
0 & \partial_{t}%
\end{array}
\right]  J\left[
\begin{array}
[c]{cc}%
\mathbf{1} & 0\\
0 & -\partial_{t}%
\end{array}
\right]  =-\mathrm{i}\partial_{t}\left[
\begin{array}
[c]{cc}%
0 & \mathrm{i}\mathbf{1}\\
\mathrm{i}\mathbf{1} & 0
\end{array}
\right]  =-\mathrm{i}\partial_{t}\tilde{J} \label{mhzne4a}%
\end{equation}
Note that the matrices $\tilde{J}$ and $\tilde{M}$ are respectively
antihermitian and hermitian, that is%
\begin{equation}
\tilde{J}^{\ast}=-\tilde{J},\qquad\tilde{M}^{\ast}=\tilde{M}. \label{mhzne4}%
\end{equation}
Notice also that the definitions of $V$ and $\tilde{J}$ in (\ref{mhzne2}%
)-(\ref{mhzne3}) imply the identity%
\begin{equation}
V^{\ast}\tilde{J}V=\mathrm{i}\left[  \mathsf{p}_{z}^{\ast}\partial
_{t}\mathsf{q}+\left(  \partial_{t}\mathsf{q}\right)  ^{\ast}\mathsf{p}%
_{z}\right]  =2\mathrm{i\operatorname{Re}}\left\{  \left(  \partial
_{t}\mathsf{q}\right)  ^{\ast}\mathsf{p}_{z}\right\}  , \label{mhzne5}%
\end{equation}
which via the theory of Hamiltonian equations can be associated with the
energy conservation law as we show below.

\subsection{Canonical and Hamilton equations\label{AppCanHam}}

In this section we provide a concise review of canonical and Hamilton
equations following \cite[II.3.1-4]{YakSta1}. By \emph{canonical} we call an
equation of the form%
\begin{equation}
\tilde{J}\frac{dz}{dt}=\tilde{H}\left(  t\right)  z, \label{Jzhz1}%
\end{equation}
where $\tilde{H}\left(  t\right)  $ is a $2n\times2n$ symmetric matrix valued
function with real entries and $\tilde{J}$ is a constant $2n\times2n$
nondegenerate skew-symmetric matrix with real entries, that is%
\begin{equation}
\tilde{H}^{\mathrm{T}}\left(  t\right)  =\tilde{H}\left(  t\right)
,\quad\tilde{J}^{\mathrm{T}}=-\tilde{J},\quad\left\vert \tilde{J}\right\vert
\neq0. \label{Jzhz2}%
\end{equation}
The matrix $\tilde{H}\left(  t\right)  $ in (\ref{Jzhz1}) is a called
"Hamiltonian" of the equation. A standard form of $2n\times2n$ nondegenerate
skew-symmetric matrix $J$ \ is%
\begin{equation}
J_{2n}=\left[
\begin{array}
[c]{cc}%
0 & -\mathbf{1}_{n}\\
\mathbf{1}_{n} & 0
\end{array}
\right]  . \label{Jzhz3}%
\end{equation}
The canonical equation (\ref{Jzhz1}) can be always reduced to the special form%
\begin{equation}
J_{2n}\frac{dx}{dt}=H\left(  t\right)  x, \label{Jzhz4}%
\end{equation}
by means of a linear change of variables,\textit{\ i.e.}
\begin{equation}
x=Sz,\qquad\tilde{J}=S^{\mathrm{T}}J_{2n}S,\qquad\tilde{H}\left(  t\right)
=S^{\mathrm{T}}H\left(  t\right)  S \label{Jzhz5}%
\end{equation}
for some real nondegenerate $2n\times2n$ matrix $S$.

We call an equation \emph{Hamiltonian} if it is of the form (\ref{Jzhz1}) and
(i) $\tilde{H}\left(  t\right)  $ is a Hermitian matrix with complex valued
entries; (ii) $\tilde{J}$ is a constant nondegenerate antihermitian matrix,
that is%
\begin{equation}
\tilde{H}^{\ast}\left(  t\right)  =\tilde{H}\left(  t\right)  ,\qquad\tilde
{J}^{\ast}=-\tilde{J},\qquad\left\vert \tilde{J}\right\vert \neq0.
\label{Jzhz6}%
\end{equation}
Canonical equations are of course Hamiltonian. A Hamiltonian equation
(\ref{Jzhz1}) can be always reduced by a transformation $x=Sz$ \ with a
nondegenerate $S$ \ to the following special form%
\begin{equation}
-iG_{0}\frac{dx}{dt}=H_{0}\left(  t\right)  x, \label{Jzhz7}%
\end{equation}
where $H_{0}\left(  t\right)  $ is a Hermitian matrix and%
\begin{equation}
G_{0}=\left[
\begin{array}
[c]{cc}%
\mathbf{1}_{p} & 0\\
0 & -\mathbf{1}_{q}%
\end{array}
\right]  ,\text{ where }p+q=2n. \label{Jzhz8}%
\end{equation}

Any matrix solution $Z\left(  t\right)  $ to the Hamiltonian equation
(\ref{Jzhz1}) satisfies the identity, \cite[II.3.4]{YakSta1}
\begin{equation}
Z\left(  t\right)  ^{\ast}\tilde{J}Z\left(  t\right)  =\tilde{J},
\label{Jzhz9}%
\end{equation}
and for any two vector solutions $z_{1}\left(  t\right)  $ and $z_{2}\left(
t\right)  $ there holds%
\begin{equation}
\left(  z_{1}\left(  t\right)  ,\tilde{J}z_{2}\left(  t\right)  \right)
=\left[  z_{1}\left(  t\right)  \right]  ^{\ast}\tilde{J}z_{2}\left(
t\right)  =\operatorname*{constant}. \label{Jzhz10}%
\end{equation}
(so called Poincar\'{e} invariant).

\subsection{Energy exchange between subsystems\label{AppEnergyExchange}}

In this section we derive a general formula for the energy flux between two
systems constituting a closed conservative system described by the Lagrangian
$\mathcal{L=L(}\mathsf{q}_{t},\mathsf{q}_{z})$ With the MTLB Lagrangian in
mind let us put $\mathsf{q}=(Q,q)$ and assume that $\mathcal{L} $ can be split
as%
\begin{equation}
\mathcal{L}=\mathcal{L}_{1}\left(  \partial_{t}Q,\partial_{;z}Q\right)
+\mathcal{L}_{2}\left(  \partial_{t}q,\partial_{z}q\right)  , \label{gblag1}%
\end{equation}
where
\[
\partial_{;z}Q=\partial_{z}Q+B\partial_{z}q
\]
and $B$ is a fixed matrix. \ The Lagrangian $\mathcal{L}$ of the general form
(\ref{gblag1}) describes two coupled interacting systems. The special form of
coupling via the modified derivative $\partial_{;z}Q$ in (\ref{gblag1})
resembles the minimal coupling in the charge gauge theory. The variable $q$
plays the role of the gauge field potential and $B$ \ plays the role of
coupling constant.

The corresponding Euler-Lagrange equations are (\ref{gblag3}), (\ref{gblag4})%
\begin{equation}
\partial_{t}\frac{\partial\mathcal{L}_{1}}{\partial\partial_{t}Q}+\partial
_{z}\frac{\partial\mathcal{L}_{1}}{\partial\partial_{;z}Q}=0, \label{gblag3}%
\end{equation}%
\begin{equation}
\partial_{t}\frac{\partial\mathcal{L}_{2}}{\partial\partial_{t}q}+\partial
_{z}\left[  \frac{\partial\mathcal{L}_{2}}{\partial\partial_{z}q}%
+\frac{\partial\mathcal{L}_{1}}{\partial\partial_{;z}Q}B\right]  =0,
\label{gblag4}%
\end{equation}
where the derivative $\frac{\partial\mathcal{L}}{\partial Q}$ of the scalar
function $\mathcal{L}$ with respect to a column-vector $Q$ is understood as a
row-vector of the same dimension.

Recall now that the energy conservation law for the entire system has the
form, \cite[38.2-3]{GelFom}, \cite[13.7]{Gold}%
\begin{equation}
\partial_{t}H+\partial_{z}S=0, \label{gblag5}%
\end{equation}
where $H$ and $S$ are the energy and energy flux densities defined by%
\begin{equation}
H=H_{1}+H_{2},\qquad S=S_{1}+S_{2}, \label{gblag6}%
\end{equation}
with the following expressions for the individual energies and energy fluxes%
\begin{equation}
H_{1}=\frac{\partial\mathcal{L}_{1}}{\partial\partial_{t}Q}\partial
_{t}Q-\mathcal{L}_{1}\left(  \partial_{t}Q,\partial_{;z}Q\right)  ,\qquad
S_{1}=\frac{\partial\mathcal{L}_{1}}{\partial\partial_{;z}Q}\partial_{t}Q,
\label{gblag7}%
\end{equation}%
\begin{equation}
H_{2}=\frac{\partial\mathcal{L}_{2}}{\partial\partial_{t}q}\partial
_{t}q-\mathcal{L}_{2}\left(  \partial_{t}q,\partial_{z}q\right)  ,\qquad
S_{2}=\left[  \frac{\partial\mathcal{L}_{2}}{\partial\partial_{z}q}%
+\frac{\partial\mathcal{L}_{1}}{\partial\partial_{;z}Q}B\right]  \partial
_{t}q. \label{gbalg8}%
\end{equation}
The above expressions imply the following identities for the first system%
\begin{gather}
\partial_{t}H_{1}=\frac{\partial\mathcal{L}_{1}}{\partial\partial_{t}%
Q}\partial_{t}^{2}Q+\partial_{t}\left(  \frac{\partial\mathcal{L}_{1}%
}{\partial\partial_{t}Q}\right)  \partial_{t}Q-\frac{\partial\mathcal{L}_{1}%
}{\partial\partial_{t}Q}\partial_{t}^{2}Q-\frac{\partial\mathcal{L}_{1}%
}{\partial\partial_{;z}Q}\left(  \partial_{tz}^{2}Q+B\partial_{tz}%
^{2}q\right)  =\label{gblag9}\\
=\partial_{t}\left(  \frac{\partial\mathcal{L}_{1}}{\partial\partial_{t}%
Q}\right)  \partial_{t}Q-\frac{\partial\mathcal{L}_{1}}{\partial\partial
_{;z}Q}\left(  \partial_{tz}^{2}Q+B\partial_{tz}^{2}q\right)  ,\nonumber
\end{gather}%
\begin{equation}
\partial_{z}S_{1}=\partial_{z}\left(  \frac{\partial\mathcal{L}_{1}}%
{\partial\partial_{;z}Q}\right)  \partial_{t}Q+\frac{\partial\mathcal{L}_{1}%
}{\partial\partial_{;z}Q}\partial_{tz}^{2}Q. \label{gblag10}%
\end{equation}
The equations (\ref{gblag9}), (\ref{gblag10}), combined with the
Euler-Lagrange equations (\ref{gblag3}), yield the following energy
conservation law for the first system%
\begin{equation}
\partial_{t}H_{1}+\partial_{z}S_{1}=-\frac{\partial\mathcal{L}_{1}}%
{\partial\partial_{;z}Q}B\partial_{tz}^{2}q, \label{gblag11}%
\end{equation}
where the right-hand side of (\ref{gblag11}) can be interpreted as the power
flow density from the second system into the first one.

Carrying out similar computations for the second system we obtain%
\begin{gather}
\partial_{t}H_{2}=\frac{\partial\mathcal{L}_{2}}{\partial\partial_{t}%
q}\partial_{t}^{2}q+\partial_{t}\left(  \frac{\partial\mathcal{L}_{2}%
}{\partial\partial_{t}q}\right)  \partial_{t}q-\frac{\partial\mathcal{L}_{2}%
}{\partial\partial_{t}q}\partial_{t}^{2}q-\frac{\partial\mathcal{L}_{2}%
}{\partial\partial_{z}q}\partial_{tz}^{2}q=\label{gblag12}\\
=\partial_{t}\left(  \frac{\partial\mathcal{L}_{2}}{\partial\partial_{t}%
q}\right)  \partial_{t}q-\frac{\partial\mathcal{L}_{2}}{\partial\partial_{z}%
q}\partial_{tz}^{2}q,\nonumber
\end{gather}%
\begin{equation}
\partial_{z}S_{2}=\partial_{z}\left[  \frac{\partial\mathcal{L}_{2}}%
{\partial\partial_{z}q}+\frac{\partial\mathcal{L}_{1}}{\partial\partial_{;z}%
Q}B\right]  \partial_{t}q+\left[  \frac{\partial\mathcal{L}_{2}}%
{\partial\partial_{z}q}+\frac{\partial\mathcal{L}_{1}}{\partial\partial_{;z}%
Q}B\right]  \partial_{tz}^{2}q. \label{gblag13}%
\end{equation}
Combining equations (\ref{gblag12}) and (\ref{gblag13}) with the
Euler-Lagrange equations (\ref{gblag4}) for the second system we obtain the
following conservation law%
\begin{equation}
\partial_{t}H_{2}+\partial_{z}S_{2}=\frac{\partial\mathcal{L}_{1}}%
{\partial\partial_{;z}Q}B\partial_{tz}^{2}q, \label{gblag14}%
\end{equation}
where the right-hand side of (\ref{gblag14}) can be interpreted as the power
density flow transferred from the first system into the second one.

Notice that relations (\ref{gblag11}) and (\ref{gblag14}) have right-hand
sides of the same magnitude and opposite signs. This can be viewed as a
manifestation of the conservation of energy for the entire system. Indeed we
recover (\ref{gblag5}) by adding (\ref{gblag11}) and (\ref{gblag14}).

\subsection{Amplification for homogeneous MTLB systems:
proofs.\label{AppAmplification}}

This section contains rigorous formulations and proofs of the assertions made
in Section \ref{AmplMTL-beam}.

\begin{theorem}
Let $\ $the hypotheses in Theorem \ref{TeoremAmplification} hold. Then, there
is a unique pair of complex conjugate solutions $v_{0},$ $v_{0}^{\ast}$ of the
equation $\left\vert \widetilde{A}(v)\right\vert =0$, where $\widetilde{A}(v)$
is defined in (\ref{ddcj1}), (\ref{ddcj2}).
\end{theorem}

\begin{proof}
By our assumption, the equation $\left\vert A(v)\right\vert =\left\vert
-v^{2}L+C^{-1}\right\vert =0$ has exactly $2n$ real roots, $\pm v_{1},\pm
v_{2},...\pm v_{n},$ with $v_{i}>0$ $(\lambda_{i}=v_{i}^{2})$. We assume in
what follows that they are ordered: $0<v_{1}\leq v_{2}\leq...\leq v_{n}$ and
each root is repeated a number of times equal to its multiplicity. \ If
$\left\vert A(v)\right\vert \neq0,$ the following decomposition holds%
\begin{equation}
\left\vert \widetilde{A}(v)\right\vert =\left\vert A(v)\right\vert \left[
d-\xi(v-u_{0})^{2}-D^{T}(A(v))^{-1}D\right]  . \label{CanonicalFact}%
\end{equation}
This follows from the following more general fact: if $M$ is a square block
matrix of the form%
\[
M=\left[
\begin{array}
[c]{cc}%
A_{1} & A_{2}\\
A_{3} & A_{4}%
\end{array}
\right]  ,
\]
where $A_{1},A_{4}$ are square matrices with $\left\vert A_{1}\right\vert
\neq0$, then%
\[
\left\vert M\right\vert =\left\vert A_{1}\right\vert \left\vert A_{4}%
-A_{3}A_{1}^{-1}A_{2}\right\vert ,
\]
see e.g. \cite[Lemma 2.8.6, page 108]{Bern}. Observe that in our case
$A_{2}=D$ is a column matrix and $A_{3}=D^{\mathrm{T}}$ is a row matrix. Then,
if $\left\vert A(v)\right\vert \neq0$, $v$ is a root of $\left\vert
\widetilde{A}(v)\right\vert =0$ if and only if it is a root of the equation
\[
-\xi(v-u_{0})^{2}=D^{T}(A(v))^{-1}D-d=:R(v).
\]
The function $R(v)$ above turns out to have very nice properties. A well known
fact from linear algebra concerning simultaneous diagonalization of two
quadratic forms, one of which is positive, assures that there exists a
non-degenerate matrix $P$ such that%
\[
P^{T}A(v)P=\mathrm{diag}\text{ }(v):=\left[
\begin{array}
[c]{cccc}%
v_{1}^{2}-v^{2} & 0 & \cdots & 0\\
0 & v_{2}^{2}-v^{2} & \cdots & 0\\
\vdots & \vdots & \ddots & 0\\
0 & 0 & 0 & v_{n}^{2}-v^{2}%
\end{array}
\right]  .
\]
Consequently,%
\[
D^{T}(A(v))^{-1}D=\widetilde{D}^{T}\left[
\begin{array}
[c]{cccc}%
\frac{1}{v_{1}^{2}-v^{2}} & 0 & \cdots & 0\\
0 & \frac{1}{v_{2}^{2}-v^{2}} & \cdots & 0\\
\vdots & \vdots & \ddots & 0\\
0 & 0 & 0 & \frac{1}{v_{n}^{2}-v^{2}}%
\end{array}
\right]  \widetilde{D}=%
%TCIMACRO{\dsum \limits_{1}^{n}}%
%BeginExpansion
{\displaystyle\sum\limits_{1}^{n}}
%EndExpansion
\frac{\widetilde{D}_{i}^{2}}{v_{i}^{2}-v^{2}},
\]
where $\widetilde{D}=P^{T}D$. $\ $Therefore,%
\[
R(v)=%
%TCIMACRO{\dsum \limits_{1}^{n}}%
%BeginExpansion
{\displaystyle\sum\limits_{1}^{n}}
%EndExpansion
\frac{\widetilde{D}_{i}^{2}}{v_{i}^{2}-v^{2}}-d
\]
is a rational function defined on the set $\left\{  v:\text{ }\left\vert
A(v)\right\vert \neq0\right\}  $. It is immediately seen that $R$ is an even
function, exhibiting vertical asymptotes at $v=\pm v_{i}$ if at least one of
the $\widetilde{D}_{k}$ associated to $v_{i}$ does not vanish ($v_{i}$ may be
a multiple root). For $v>0,$ each branch between two consecutive asymptotes is
increasing and they are decreasing for $v<0.$ Moreover, $\lim_{v\rightarrow
\infty}R(v)=-d$ . Also,%
\begin{gather*}
R(0)+d=D^{T}(A(0))^{-1}D=D^{T}CD=%
%TCIMACRO{\dsum \limits_{i=1}^{n}}%
%BeginExpansion
{\displaystyle\sum\limits_{i=1}^{n}}
%EndExpansion%
%TCIMACRO{\dsum \limits_{j=1}^{n}}%
%BeginExpansion
{\displaystyle\sum\limits_{j=1}^{n}}
%EndExpansion
C_{ij}D_{i}D_{j}=\\
=%
%TCIMACRO{\dsum \limits_{i=1}^{n}}%
%BeginExpansion
{\displaystyle\sum\limits_{i=1}^{n}}
%EndExpansion%
%TCIMACRO{\dsum \limits_{j=1}^{n}}%
%BeginExpansion
{\displaystyle\sum\limits_{j=1}^{n}}
%EndExpansion
C_{ij}\left[
%TCIMACRO{\dsum \limits_{k=1}^{n}}%
%BeginExpansion
{\displaystyle\sum\limits_{k=1}^{n}}
%EndExpansion
(C^{-1})_{ik}\right]  \left[
%TCIMACRO{\dsum \limits_{r=1}^{n}}%
%BeginExpansion
{\displaystyle\sum\limits_{r=1}^{n}}
%EndExpansion
(C^{-1})_{jr}\right]  =\\
=%
%TCIMACRO{\dsum \limits_{k=1}^{n}}%
%BeginExpansion
{\displaystyle\sum\limits_{k=1}^{n}}
%EndExpansion%
%TCIMACRO{\dsum \limits_{r=1}^{n}}%
%BeginExpansion
{\displaystyle\sum\limits_{r=1}^{n}}
%EndExpansion
\left[
%TCIMACRO{\dsum \limits_{i=1}^{n}}%
%BeginExpansion
{\displaystyle\sum\limits_{i=1}^{n}}
%EndExpansion
(C^{-1})_{ik}%
%TCIMACRO{\dsum \limits_{j=1}^{n}}%
%BeginExpansion
{\displaystyle\sum\limits_{j=1}^{n}}
%EndExpansion
C_{ij}(C^{-1})_{jr}\right]  =%
%TCIMACRO{\dsum \limits_{k=1}^{n}}%
%BeginExpansion
{\displaystyle\sum\limits_{k=1}^{n}}
%EndExpansion%
%TCIMACRO{\dsum \limits_{r=1}^{n}}%
%BeginExpansion
{\displaystyle\sum\limits_{r=1}^{n}}
%EndExpansion
\left[
%TCIMACRO{\dsum \limits_{i=1}^{n}}%
%BeginExpansion
{\displaystyle\sum\limits_{i=1}^{n}}
%EndExpansion
(C^{-1})_{ik}\delta_{ir}\right] \\
=%
%TCIMACRO{\dsum \limits_{k=1}^{n}}%
%BeginExpansion
{\displaystyle\sum\limits_{k=1}^{n}}
%EndExpansion%
%TCIMACRO{\dsum \limits_{r=1}^{n}}%
%BeginExpansion
{\displaystyle\sum\limits_{r=1}^{n}}
%EndExpansion
(C^{-1})_{rk}=%
%TCIMACRO{\dsum \limits_{k=1}^{n}}%
%BeginExpansion
{\displaystyle\sum\limits_{k=1}^{n}}
%EndExpansion
D_{k}=d;
\end{gather*}
hence $R(0)=0.$ Since $C^{-1}$ is non-degenerate, $D\neq0$. Moreover, since
the matrix $P$ is non-degenerate, we have $\widetilde{D}\neq0.$ Therefore, the
graph has at least two vertical asymptotes and always exhibits a central
symmetric branch with the minimum at the point $(0,0).$ The number of real
roots of the equation%
\[
-\xi(v-u_{0})^{2}=R(v)
\]
is the number of intersection points of the parabola $y=f(v):=-\xi
(v-u_{0})^{2}$ and the graph of $R.$ For $\xi$ small, it is exactly the number
of monotonic branches (all branches, except for the central one), which
coincides with the number of asymptotes. This number is always between $2$ and
$2n,$ depending on the number of vanishing $\widetilde{D}_{i}$ and on the
possible multiple roots; a precise description is given below. Moreover, it is
easily seen that whenever $u_{0}\in(0,v_{1}],$ the number of intersection
points is equal to the number of asymptotes irrespective of the value of
$\xi>0$; see Figure \ref{FigAmpl} (a), whereas $\xi$ small is needed
otherwise; indeed, in Figure \ref{FigAmpl} (b) a large value of $\xi$ produces
three points of intersection with the far right branch of the graph of $R$,
making the total number of intersection points exceed by two the number of
asymptotes. If either (i) or (ii) holds, the intersections are transversal,
hence the roots are simple. The previous assertions follow easily and
rigorously from the monotonicity properties of $R$ and $f$ \ but their clear
geometric meaning makes a lengthy proof unnecessary.

So far, we have considered the real roots of the equation$\left\vert
\widetilde{A}(v)\right\vert =0$ in the set $\left\{  v:\det A(v)\neq0\right\}
$. Next, we consider the possible roots of the equation in the complementary
set $\left\{  \pm v_{1},\pm v_{2},...\pm v_{n}\right\}  $. Multiplying the
matrix $\widetilde{A}(v)$ by $\widehat{P}^{T}$ from the left and by
$\widehat{P}$ from the right, where%
\[
\widehat{P}=\left[
\begin{array}
[c]{cc}%
P & 0\\
0 & 1
\end{array}
\right]  ,
\]
there follows that the equation $\ \left\vert \widetilde{A}(v)\right\vert =0$
is equivalent to the equation%
\[
\Delta(v):=\left\vert
\begin{array}
[c]{cc}%
\begin{array}
[c]{cccc}%
v_{1}^{2}-v^{2} & 0 & .. & 0\\
0 & v_{2}^{2}-v^{2} & .. & \vdots\\
\vdots & \vdots & .. & 0\\
0 & .. & 0 & v_{n}^{2}-v^{2}%
\end{array}
& \widetilde{D}\\
\widetilde{D}^{T} & d-\xi(v-u_{0})^{2}%
\end{array}
\right\vert =0,
\]
where, as before, $\widetilde{D}=P^{T}D.$ Let us analyze under what condition
$\pm v_{i}$ are roots of the equation $\Delta(v)=0$. Expanding the determinant
with respect to the last column, and then the $n$-th order minor corresponding
to $\widetilde{D}_{i}$ with respect to its $i$-th column, we get the
expression%
\begin{gather}
\Delta(v)=-\widetilde{D}_{1}^{2}\left\vert
\begin{array}
[c]{cccc}%
v_{2}^{2}-v^{2} & 0 & .. & 0\\
0 & v_{3}^{2}-v^{2} & .. & \vdots\\
\vdots & \vdots & .. & 0\\
0 & .. & 0 & v_{n}^{2}-v^{2}%
\end{array}
\right\vert -\widetilde{D}_{2}^{2}\left\vert
\begin{array}
[c]{cccc}%
v_{1}^{2}-v^{2} & 0 & .. & 0\\
0 & v_{3}^{2}-v^{2} & .. & \vdots\\
\vdots & \vdots & .. & 0\\
0 & .. & 0 & v_{n}^{2}-v^{2}%
\end{array}
\right\vert -...\label{BigDeterminant}\\
-\widetilde{D}_{n}^{2}\left\vert
\begin{array}
[c]{cccc}%
v_{1}^{2}-v^{2} & 0 & .. & 0\\
0 & v_{2}^{2}-v^{2} & .. & \vdots\\
\vdots & \vdots & .. & 0\\
0 & \cdots & 0 & v_{n-1}^{2}-v^{2}%
\end{array}
\right\vert +\left[  d-\xi(v-u_{0})^{2}\right]  \left\vert
\begin{array}
[c]{cccc}%
v_{1}^{2}-v^{2} & 0 & .. & 0\\
0 & v_{2}^{2}-v^{2} & .. & \vdots\\
\vdots & \vdots & .. & 0\\
0 & .. & 0 & v_{n}^{2}-v^{2}%
\end{array}
\right\vert ,\nonumber
\end{gather}
that is,%
\begin{equation}
\Delta(v)=%
%TCIMACRO{\dprod \limits_{i=1}^{n}}%
%BeginExpansion
{\displaystyle\prod\limits_{i=1}^{n}}
%EndExpansion
(v_{i}^{2}-v^{2})\left[  d-\xi(v-u_{0})^{2}\right]  -\sum_{i=1}^{n}%
\widetilde{D_{i}}^{2}%
%TCIMACRO{\dprod \limits_{j\neq i}}%
%BeginExpansion
{\displaystyle\prod\limits_{j\neq i}}
%EndExpansion
(v_{j}^{2}-v^{2}). \label{BigDetermBIS}%
\end{equation}
We note in passing that the factorization (\ref{CanonicalFact}) is easily
obtained from the above expression by extracting the factor
\[
\left\vert A(v)\right\vert =%
%TCIMACRO{\dprod \limits_{i=1}^{n}}%
%BeginExpansion
{\displaystyle\prod\limits_{i=1}^{n}}
%EndExpansion
(v_{i}^{2}-v^{2})
\]
under the assumption $v\neq v_{i}.$

Assume first that $\pm v_{i}$ are simple roots of $\left\vert A(v)\right\vert
=0,$ that is, that the binomial $v_{i}^{2}-v^{2}$ appears only once in the
matrix $\mathrm{diag}(v)$. Then, there follows that $\Delta(v_{i})=0$ if and
only if $\widetilde{D}_{i}=0$. Whenever this condition holds, the partial
fraction $\widetilde{D}_{i}^{2}/(v_{i}^{2}-v^{2})$ in the expression of $R$
disappears and the number of asymptotes is reduced by two. The number of real
roots is thus increased by two ( $\pm v_{i})$ and reduced by two, leaving the
total number of roots unaffected.

Let us next consider the case of a multiple root. Assume that $v_{i}%
=v_{i+1}=...=v_{i+k-1}$, hence the binomial $v_{i}^{2}-v^{2}$ appears $k$
times in $\mathrm{diag}(v)$, $k>1$. Then $\pm v_{i}$ are necessarily roots of
$\Delta(v)=0,$ as it can be readily seen from (\ref{BigDetermBIS}). As for
their multiplicity, there are two cases:

\begin{enumerate}
\item[a)] multiplicity $=k,$ if all of $\widetilde{D}_{i},\widetilde{D}%
_{i+1},...\widetilde{D}_{i+k-1}$ are zero, since in this case all non-zero
terms in (\ref{BigDetermBIS}) contain $k$ times the factor $v_{i}^{2}-v^{2}$;

\item[b)] multiplicity $=k-1,$ if not all of $\widetilde{D}_{i},\widetilde{D}%
_{i+1},...\widetilde{D}_{i+k-1}$ are zero, since in this case the terms in
(\ref{BigDetermBIS}) corresponding to the non-zero $\widetilde{D}_{r}$ \ with
$r\in\left\{  i,i+1,...i+k-1\right\}  $ contain the factor $v_{i}^{2}-v^{2}$
only $k-1$ times, while the rest contain it $k$ times.
\end{enumerate}

In case (a), all the fractions with denominator $v_{i}^{2}-v^{2}$ are missing
in the rational function $R,$ with consequent reduction of the number of
asymptotes (with respect to the total possible number $2n$) by $2k$, which is
precisely the number of additional roots, counting their multiplicity. Thus
the total number of real roots is unaffected.

In case (b), there is one fraction with denominator $v_{i}^{2}-v^{2}.$Thus the
total number of asymptotes is reduced by $2k-2,$ which is the number of
additional roots, counting their multiplicity.

Summing up, the total number of real roots of $\Delta(v)=0$ (counting their
multiplicity) is exactly $2n$ under our assumptions. Since the total number of
roots of $\Delta(v)=0$ is $2n+2,$ there is necessarily one and only one pair
of complex conjugate roots, thus proving the assertion.
\end{proof}

The following theorem deals with the behavior of amplification as
$\xi\rightarrow0$ and as $\xi\rightarrow\infty$.

\begin{theorem}
Let $v_{0},\overline{v_{0}}$ with $\operatorname*{Im}v_{0}>0$ denote the
unique pair of complex conjugate roots of the equation $\left\vert
\widetilde{A}(v)\right\vert =0$ under the assumptions of Theorem
\ref{TeoremAmplification}. Let $k_{0}=\omega/v_{0}.$ Then,%
\begin{equation}
-\operatorname*{Im}k_{0}\sim\frac{C}{\sqrt{\xi}}\text{ \ as }\xi
\rightarrow0,\text{ }C>0. \label{AsBehxitozero}%
\end{equation}
Under the additional assumption $u_{0}=v_{1}$ we also have%
\begin{equation}
-\operatorname*{Im}k_{0}\sim\frac{C^{\prime}}{\sqrt[3]{\xi}}\text{ \ as }%
\xi\rightarrow\infty,\text{ }C^{\prime}>0. \label{AsBehxitoinfty}%
\end{equation}

\end{theorem}

\begin{proof}
The idea of the proof is to use very detailed information about real roots in
combination with well known Vieta's formulas relating the roots to the
coefficients of the corresponding polynomial. Let us first prove
(\ref{AsBehxitozero}). Denote the $2n$ real roots of the equation $\left\vert
\widetilde{A}(v)\right\vert =0$ by $v_{1}^{+}$, $v_{2}^{+},...$ $v_{n}^{+}$ ;
$-v_{1}^{-},-v_{2}^{-},...-v_{n}^{-}$, where $v_{i}^{+},v_{i}^{-}>0$ and
$0<v_{1}^{+}\leq v_{2}^{+}\leq...\leq$ $v_{n}^{+}$, $0<v_{1}^{-}\leq v_{2}%
^{-}\leq...\leq v_{n}^{-}$.$\ $The roots are repeated according to their
multiplicity and some of them may coincide with some $v_{i}$; see the proof of
Theorem \ref{TeoremAmplification}. If $n>1,$ the roots $v_{i}^{+}$and
$-v_{i}^{-}$ with $i=1,2,...(n-1)$ \ lie in the interval $\left[  -v_{n}%
,v_{n}\right]  $ for any value of \ $\xi>0$ \ (recall that by $v_{i}$ we
denote the characteristic velocities of the MTL), whereas $v_{n}^{+}$ and
$-v_{n}^{-},$ which correspond to the points of intersection of the parabola
$y=-\xi(v-u_{0})^{2}$ with the farthest right and farthest left branches of
$y=R(v),$ lie outside of this very interval.

The extreme roots $v_{n}^{+}$ and $-v_{n}^{-}$ approach $+\infty$
(respectively $-\infty$) as $\xi\rightarrow0.$ This can be proved as follows:
the parabola $y=-\xi(v-u_{0})^{2}$ is decreasing for $v>u_{0},$ its
intersection with the horizontal asymptote of $R,$ $y=-d,$ is $v^{\ast}%
=u_{0}+\sqrt{d/\xi}$ and $R(v)<-d$ for $v>v_{n}.$ Therefore, $v_{n}%
^{+}>v^{\ast}\rightarrow+\infty$ as $\xi\rightarrow0.$ A similar argument can
be applied to $-v_{n}^{-}.$ In order to establish the asymptotic behavior of
$\operatorname{Im}v_{0}$ we will make use of Vieta's formulas, relating the
roots of a polynomial to its coefficients.

We start by observing that $\left\vert \widetilde{A}(v)\right\vert $ is a
polynomial in $v$ of degree $2n+2:$%
\[
\left\vert \widetilde{A}(v)\right\vert =a_{2n+2}v^{2n+2}+a_{2n+1}%
v^{2n+1}+...a_{1}v+a_{0.}%
\]
The coefficients $a_{2n+2},a_{2n+1}$ and $a_{0}$ can be easily computed in
terms of the parameters. Indeed, $a_{0}=\left\vert \widetilde{A}(0)\right\vert
,$ which can be computed by adding the first $n$ rows of $\widetilde{A}(0)$
and subtracting the result from the last. Recalling that $D_{i}=%
%TCIMACRO{\dsum \limits_{j}}%
%BeginExpansion
{\displaystyle\sum\limits_{j}}
%EndExpansion
(C^{-1})_{ij}$ and that $d=%
%TCIMACRO{\dsum _{i}}%
%BeginExpansion
{\displaystyle\sum_{i}}
%EndExpansion
D_{i},$ we obtain%
\[
a_{0}=\left\vert \widetilde{A}(0)\right\vert =\left\vert
\begin{array}
[c]{cc}%
C^{-1} & D\\
0 & -\xi u_{0}^{2}%
\end{array}
\right\vert =-\xi u_{0}^{2}\left\vert C^{-1}\right\vert .
\]
The only addends in $\left\vert \widetilde{A}(v)\right\vert $ yielding powers
$v^{2n+2}$ or $v^{2n+1}$ are those coming from the product $\left\vert
-v^{2}L+C^{-1}\right\vert \left[  d-\xi(v-u_{0})^{2}\right]  $. $\ $Clearly,
the relevant terms are%
\[
(-1)^{n}\left\vert L\right\vert v^{2n}\left[  d-\xi(v-u_{0})^{2}\right]
+...=(-1)^{n+1}\xi\left\vert L\right\vert v^{2n+2}+2(-1)^{n}\xi u_{0}%
\left\vert L\right\vert v^{2n+1}+...
\]
where the dots stand for lower order in $v$ terms. Consequently,%
\[
a_{2n+2}=(-1)^{n+1}\xi\left\vert L\right\vert ;\text{ \ \ \ \ \ }%
a_{2n+1}=2(-1)^{n}\xi u_{0}\left\vert L\right\vert .
\]
Vieta's formulas then imply%
\begin{align}
2\operatorname*{Re}v_{0}+%
%TCIMACRO{\dsum \limits_{i=1}^{n}}%
%BeginExpansion
{\displaystyle\sum\limits_{i=1}^{n}}
%EndExpansion
v_{i}^{+}-%
%TCIMACRO{\dsum \limits_{i=1}^{n}}%
%BeginExpansion
{\displaystyle\sum\limits_{i=1}^{n}}
%EndExpansion
v_{i}^{-}  &  =-\frac{a_{2n+1}}{a_{2n+2}}=2u_{0}\label{Vieta1}\\
(-1)^{n}\left\vert v_{0}\right\vert ^{2}%
%TCIMACRO{\dprod \limits_{i=1}^{n}}%
%BeginExpansion
{\displaystyle\prod\limits_{i=1}^{n}}
%EndExpansion
v_{i}^{+}v_{i}^{-}  &  =\frac{a_{0}}{a_{2n+2}}=(-1)^{n}\frac{u_{0}%
^{2}\left\vert C^{-1}\right\vert }{\left\vert L\right\vert }=(-1)^{n}%
\frac{u_{0}^{2}}{\left\vert LC\right\vert } \label{Vieta2}%
\end{align}
Next, we study the behavior as $\xi\rightarrow0$ of both the sum and the
product of the real roots. In the asymptotic formulas below, $K_{1}%
,K_{2},K_{1}^{\prime},K_{2}^{\prime}$ etc. denote positive constants depending
on $L,C,u_{0}$ but not on $\xi$.

Let $n>1$ and suppose that the graph of $R$ has more than two asymptotes.
First of all, we note that, as $\xi\rightarrow0,$ the parabola becomes flat
and the roots $v_{i}^{+},-v_{i}^{-}$ with $i=1,2,...n-1$ become symmetric due
to the symmetry of the graph of $R$. More precisely, if we denote by
$\widehat{v}_{k}^{+},-\widehat{v}_{k}^{-}$ with $k\in\left\{
1,2,...n-1\right\}  $ the abscissas of the points on the $k$-th right
(respectively, $k$-th left) branch of the graph of $R$ for which
$R(\widehat{v}_{k}^{+})=R($ $-\widehat{v}_{k}^{-})=0,$ then clearly $v_{k}%
^{+}\left(  \xi\right)  \rightarrow\widehat{v}_{k}^{+},v_{k}^{-}\left(
\xi\right)  \rightarrow\widehat{v}_{k}^{-}$ and $\widehat{v}_{k}%
^{+}=\widehat{v_{k}}^{-}.$ Moreover, since the branches of $R$ are strictly
increasing for $v>0$ and strictly decreasing for $v<0$, $v_{k}^{+}\left(
\xi\right)  -\widehat{v}_{k}^{+}\sim-A_{k}\xi,$ $-v_{k}^{-}\left(  \xi\right)
+\widehat{v}_{k}^{-}\sim B_{k}\xi$ \ as $\xi\rightarrow0$, with $A_{k}%
,B_{k}>0.$We also note the following fact, which is used in the proof of
Section \ref{SubsEnergyExchange} and \ is a simple consequence of the lack of
symmetry of the parabola $y=-\xi(v-u_{0})^{2}$ with respect to the vertical
axis: if $v_{k}^{+},-v_{k}^{-}$ is a pair of real roots not belonging to the
set $\left\{  \pm v_{1},\pm v_{2},...\pm v_{n}\right\}  $ (and there is at
least one such pair, see the proof of Theorem \ref{TeoremAmplification}), then%
\begin{equation}
v_{k}^{+}\left(  \xi\right)  -v_{k}^{-}\left(  \xi\right)  >0.
\label{Asymmetry}%
\end{equation}
Thus in particular $B_{k}>A_{k}$ in the above asymptotic relations.This
inequality can be easily seen on the graph and given a simple analytical proof.

The roots belonging to the set $\left\{  \pm v_{1},\pm v_{2},...\pm
v_{n}\right\}  $ are symmetric and do not contribute to their sum. Therefore,%
\begin{equation}%
%TCIMACRO{\dsum \limits_{i=1}^{n-1}}%
%BeginExpansion
{\displaystyle\sum\limits_{i=1}^{n-1}}
%EndExpansion
v_{i}^{+}(\xi)-%
%TCIMACRO{\dsum \limits_{i=1}^{n-1}}%
%BeginExpansion
{\displaystyle\sum\limits_{i=1}^{n-1}}
%EndExpansion
v_{i}^{-}(\xi)=K_{1}\xi+o(\xi)\text{ as }\xi\rightarrow0.
\label{SmallRootsSum}%
\end{equation}

As for the product of roots, we have%
\begin{equation}%
%TCIMACRO{\dprod \limits_{i=1}^{n-1}}%
%BeginExpansion
{\displaystyle\prod\limits_{i=1}^{n-1}}
%EndExpansion
v_{i}^{+}(\xi)v_{i}^{-}(\xi)=(-1)^{n}K_{2}+K_{3}\xi+o(\xi)\text{\ as }%
\xi\rightarrow0. \label{SmallRootsProduct}%
\end{equation}
If there are only two asymptotes, then $v_{1}=v_{2}=...=v_{n-1}$ and the
left-hand side in (\ref{SmallRootsSum}) is zero. Also, the left-hand side in
(\ref{SmallRootsProduct}) is the constant $(-1)^{n}K_{2}$ \ Thus this case can
be formally included in (\ref{SmallRootsSum}) and (\ref{SmallRootsProduct}) by
allowing $K_{1}$ and $K_{3}$ to vanish.

Let us now consider the extreme roots. As we noted, $v_{n}^{+},v_{n}%
^{-}\rightarrow+\infty.$ More precisely, since we have%
\begin{equation}
-\xi(v_{n}^{+}-u_{0})^{2}=R(v_{n}^{+})\rightarrow-d\text{ \ as }\xi
\rightarrow0, \label{asymptequality}%
\end{equation}
then necessarily $\lim_{\xi\rightarrow0}\xi(v_{n}^{+}-u_{0})^{2}=d>0$ and thus%
\begin{equation}
v_{n}^{+}(\xi)=\sqrt{\frac{d}{\xi}}+u_{0}+E(\xi),\text{ where }E(\xi)=o\left(
\sqrt{\frac{1}{\xi}}\right)  \text{ \ as }\xi\rightarrow0. \label{beh1}%
\end{equation}
We need further refinement in the asymptotics of $E(\xi)$ as $\xi
\rightarrow0.$ To this end, we recall that%
\begin{equation}
R(v)+d\sim-\frac{A}{v^{2}}\text{ \ as }v\rightarrow\infty,\text{\quad}A>0.
\label{beh3}%
\end{equation}
Replacing $v_{n}^{+}(\xi)$ in (\ref{asymptequality}) by the expression
(\ref{beh1}) and using (\ref{beh3}), we arrive at%
\[
2E(\xi)\sqrt{\frac{d}{\xi}}+E(\xi)^{2}\rightarrow A\text{ \ as }\xi
\rightarrow0,
\]
which implies%
\[
E(\xi)=K_{3}\sqrt{\xi}+o(\sqrt{\xi})\text{ \ as }\xi\rightarrow0.
\]
Summing up, we have the following asymptotic representation for $v_{n}^{+}:$%
\begin{equation}
v_{n}^{+}(\xi)=\sqrt{\frac{d}{\xi}}+u_{0}+K_{3}\sqrt{\xi}+o(\sqrt{\xi})\text{
\ as }\xi\rightarrow0. \label{behdef}%
\end{equation}
An analogous representation takes place for $v_{n}^{-}:$%
\begin{equation}
-v_{n}^{-}(\xi)=-\sqrt{\frac{d}{\xi}}+u_{0}-K_{3}\sqrt{\xi}+o(\sqrt{\xi
})\text{ \ as }\xi\rightarrow0. \label{behdefbis}%
\end{equation}
Plugging (\ref{behdef}), (\ref{behdefbis}), (\ref{SmallRootsSum}) and
(\ref{SmallRootsProduct}) into (\ref{Vieta1}) and (\ref{Vieta2}) yields
\begin{equation}
\operatorname*{Re}v_{0}=o\left(  \sqrt{\xi}\right)  \text{ };\text{
\ \ }\left\vert v_{0}\right\vert ^{2}=K_{4}\xi+o(\xi)\text{ \ as }%
\xi\rightarrow0.
\end{equation}
As a consequence,%
\[
\operatorname*{Im}v_{0}=\sqrt{K_{4}\xi+o(\xi)}=\sqrt{K_{4}}\sqrt{\xi}%
+o(\sqrt{\xi})\text{ as }\xi\rightarrow0
\]
and, finally,%
\[
-\operatorname*{Im}k_{0}=\frac{\operatorname*{Im}v_{0}}{\text{\ \ }\left\vert
v_{0}\right\vert ^{2}}\sim\frac{K_{5}}{\sqrt{\xi}\text{\ }}\text{\ as \ }%
\xi\rightarrow0,
\]
thus proving (\ref{AsBehxitozero}) for $n>1$. If $n=1$, (\ref{behdef}) and
(\ref{behdefbis}) hold \ and plugging into (\ref{Vieta1}) and (\ref{Vieta2})
yields the same result.

We turn now to the proof of (\ref{AsBehxitoinfty}), restricting ourselves to
the case of just one line$;$ the case of several lines can be handled in a
similar fashion. First of all, it is clear that $v_{1}^{+}\downarrow u_{0}$
and $-v_{1}^{-}\uparrow-u_{0}$ as $\xi\rightarrow\infty$. This can be
rigorously proved in a way, similar to the above proof of the fact that
$v_{n}^{+}\uparrow\infty,$ $v_{n}^{-}\downarrow-\infty$ as $\xi\rightarrow0$.
Put%
\begin{equation}
v_{1}^{+}(\xi)=u_{0}+G(\xi);\text{\quad}G(\xi)>0,\text{\quad}G(\xi
)\rightarrow0\text{ as }\xi\rightarrow\infty. \label{PierceRegime1}%
\end{equation}
Near $u_{0}=v_{1}$ we have%
\begin{equation}
R(v)\sim\frac{A}{u_{0}-v}\text{ as }v\rightarrow u_{0}^{+}\text{ \ with }A>0.
\label{PierceRegime2}%
\end{equation}
After use of (\ref{PierceRegime1}) and (\ref{PierceRegime2}), the equation%
\[
-\xi(v_{1}^{+}-u_{0})^{2}=R(v_{1}^{+})
\]
yields the following asymptotic relation:%
\[
\xi G(\xi)^{3}\sim A\text{\ as }\xi\rightarrow\infty,
\]
that is,%
\[
G(\xi)\sim\frac{K_{1}^{\prime}}{\sqrt[3]{\xi}}\text{ as }\xi\rightarrow
\infty.
\]
An analogous formula takes place for the negative root:%
\[
v_{1}^{-}(\xi)=u_{0}+H(\xi)\ \text{\ with }H(\xi)\sim\frac{K_{2}^{\prime}%
}{\sqrt[3]{\xi}}\text{\ as }\xi\rightarrow\infty.
\]
Therefore,%
\[
v_{1}^{+}(\xi)-v_{1}^{-}(\xi)\sim\frac{K_{3}^{\prime}}{\sqrt[3]{\xi}%
}\text{\ as }\xi\rightarrow\infty.
\]
Applying again (\ref{Vieta1}) and (\ref{Vieta2}) we obtain%
\begin{equation}
\operatorname{Re}v_{0}=u_{0}+\frac{K_{4}^{\prime}}{\sqrt[3]{\xi}}+o\left(
\frac{1}{\sqrt[3]{\xi}}\right)  ,\text{\qquad}\left\vert v_{0}\right\vert
\rightarrow1/\sqrt{LC}+\frac{K_{5}^{\prime}}{\sqrt[3]{\xi}}+o\left(  \frac
{1}{\sqrt[3]{\xi}}\right)  \text{\ as }\xi\rightarrow\infty.
\label{PierceRegime3}%
\end{equation}
Recall that $v_{1}=1/\sqrt{LC}=u_{0}.$ The last two relations imply
$\operatorname*{Im}v_{0}\sim K_{6}^{\prime}/\sqrt[3]{\xi}$. Finally,%
\[
-\operatorname*{Im}k_{0}=\frac{\operatorname*{Im}v_{0}}{\text{\ \ }\left\vert
v_{0}\right\vert ^{2}}\sim\frac{K_{7}^{\prime}}{\sqrt[3]{\xi}}\text{\qquad as
}\xi\rightarrow\infty.
\]
as was to be proved.
\end{proof}

One can also verify that if $u_{0}<1/\sqrt{LC}$ then both $\operatorname*{Im}%
v_{0}$ and $\operatorname*{Im}k_{0}$ have a finite, nonzero limit as
$\xi\rightarrow\infty.$

\textbf{Acknowledgment:} This research was supported by AFOSR MURI Grant
FA9550-12-1-0489 administered through the University of New Mexico. The
authors are grateful to F. Capolino and A. Tamma for helpful discussions.

\end{document}